\documentclass[10pt,doublesided,doublecolumn]{IEEEtran} 
\usepackage{cite}
\usepackage{graphics}
\usepackage{amssymb}
\usepackage{amscd}
\usepackage{dsfont}
\usepackage{amsmath}
\usepackage{epsfig}
\usepackage{psfrag}
\usepackage{rotating}
\usepackage{amsmath}
\usepackage{amsfonts}
\usepackage{url}
\usepackage{color}
\newtheorem{lemma}{Lemma}
\newtheorem{theorem}{Theorem}

\newtheorem{definition}[theorem]{Definition}
\newtheorem{proposition}[theorem]{Proposition}

\newcommand{\BR}{\mathrm{BR}}
\newcommand{\SNR}{ \mathrm{SNR}}

\newcommand{\bs}{\boldsymbol}
\newcommand{\mc}{\mathcal}

\newcommand{\defn}{\stackrel{\triangle}{=} }
\newcommand{\ds}{\displaystyle}

\renewcommand{\d}{\mathrm{d}}

\newcommand{\diag}{\mathrm{diag}}

\newcommand{\gameone}{\mathcal{G}_{(a)}}
\newcommand{\gametwo}{\mathcal{G}_{(b)}}
\newcommand{\gamei}{\mathcal{G}_{(i)}}

\newcommand{\channelvector}{\left(g_{11}, g_{12},g_{21},g_{22}\right)}
\newcommand{\channelset}{\left\lbrace g_{ij}\right\rbrace_{\forall (i,j) \in \mathcal{K} \times \mathcal{S}}}

\begin{document}
\title{Equilibria of Channel Selection Games in Parallel Multiple Access Channels}
\author{Samir~M.~Perlaza,
           ~and~Samson~Lasaulce,
          ~and~M\'{e}rouane~Debbah 
\thanks{S.~ M. Perlaza is with Alcatel-Lucent Chair in Flexible Radio at SUPELEC. $3$ rue Joliot-Curie, $91192$, Gif-sur-Yvette, cedex. France. (samir.medina-perlaza@supelec.fr)}%
\thanks{S. Lasaulce is with LSS (CNRS--SUPELEC--Univ. Paris Sud). $3$ rue Joliot-Curie, $91192$, Gif-sur-Yvette, Cedex. France. (lasaulce@lss.supelec.fr) }%
\thanks{M. Debbah is with Alcatel-Lucent Chair in Flexible Radio at SUPELEC. $3$ rue Joliot-Curie, $91192$, Gif-sur-Yvette, Cedex. France. (merouane.debbah@supelec.fr)}%
\thanks{Part of this work was presented at the ICST/ACM International Workshop on Game Theory in Communication Networks \cite{Perlaza-Gamecomm-09}}
}

\maketitle
\begin{abstract}

\boldmath
In this paper, we study the decentralized parallel multiple access channel (MAC) when transmitters selfishly maximize their individual spectral efficiency by selecting a single channel to transmit. More specifically, we investigate the set of Nash equilibria (NE) of decentralized networks comprising several transmitters communicating with a single receiver that implements single user decoding. 
This scenario is modeled as a one-shot game where the players (the transmitters) have discrete action sets (the channels).
We show that the corresponding game has always at least one NE in pure strategies, but, depending on certain parameters, the game might possess several NE. We provide an  upper bound for the maximum number of NE as a function of the number of transmitters and available channels.  
The main contribution of this paper is a mathematical proof of the existence of a Braess-type paradox. In particular, it is shown that under the assumption of a fully loaded network, when transmitters are allowed to use all the available channels, the corresponding sum spectral efficiency achieved at the NE is lower or equal than the sum spectral efficiency achieved when transmitters can use only one channel. A formal proof of this observation is provided in the case of small networks. For general scenarios, we provide numerical examples that show that the same effect holds as long as the network is kept fully loaded. We conclude the paper by considering the case of successive interference cancellation at the receiver. In this context, we show that the power allocation vectors at the NE are capacity maximizers. Finally, simulations are presented to verify  our theoretical results.
 \end{abstract}

\section{Introduction}\label{SecIntroduction}

In their original definition \cite{Cover-Book-91}, multiple access channels (MAC) correspond to a communication scenario where several transmitters communicate with a single receiver trough a common channel. In parallel MAC, each transmitter can exploit a common set of orthogonal (sub)channels to communicate with the receiver. Often, channel orthogonality is assumed in the frequency domain, and thus, channels can be understood as different non-overlapping frequency bands. This model allows one to study communication scenarios such as 802.11-based wireless local area networks (WLANs) \cite{Chieochan-2010, Lasaulce-2008}, 
distributed soft or hard  handovers in cellular systems \cite{Perlaza-Gamecomm-09},  or throughput-maximizing power control in multi-carrier code division multiple access (MC-CDMA) systems \cite{Meshkati2006}. 

\noindent
In this paper, we analyze the parallel MAC assuming that transmitters selfishly maximize their individual spectral efficiency (ISE)  by autonomously selecting a single channel to perform their transmission. Here, the channel selection policy is not imposed by the receiver to the transmitters, which justifies the terms decentralized parallel MAC.
The motivation for studying this scenario and, in particular, the limitation of using a single channel for transmitting, stems from the fact that this is often a practical constraint in some wireless networks, for instance, Wi-Fi networks. Moreover, the choice of enforcing radio devices to use only a subset of all the available channels has been proved to be beneficial in the case of rate-efficient centralized parallel MAC when using successive interference cancellation \cite{Perlaza-Gamecomm-09}. In the case of energy-efficient decentralized parallel MAC, it has been shown that using a single channel is a dominant strategy \cite{Meshkati2006}.
Within this framework, we study this scenario, to which we refer to as channel selection (CS) problem, using a one-shot game model. The players (the transmitters) have discrete action sets (the channels) and their utility function (performance metric) corresponds to their ISE. Our interest focuses on the analysis of the set of Nash equilibria (NE)  \cite{Nash-1950} of this game. The relevance of the NE relies, in part, on the fact that it describes a network state where the channel used by each transmitter is individually optimal with respect to the channels adopted by all the other transmitters in the network. Another reason is that an NE can be reached in a fully decentralized fashion when radio devices interact during certain time following particular behavioral rules \cite{Rose-Perlaza-CommMag-2011, Lasaulce-Tembine-Book-2011}.   

\noindent
We distinguish the CS game described above from the power allocation (PA) game. In the PA game, transmitters can simultaneously use all the available channels, and thus, the set of actions is a convex and closed set  \cite{Mertikopoulos-JSAC-2011}. Indeed, the PA problem in decentralized parallel MAC has been well investigated in the wireless literature \cite{Belmega-TWC-2009, Mertikopoulos-JSAC-2011, ElGamal-08, He-10, Yu-2004}. In these works, the main contribution consists in conducting a complete analysis of the set of NE of the corresponding one-shot games. Nonetheless, very few is known about the set of equilibria in decentralized parallel MAC, with CS policies. As we shall see, by dropping the convexity of the set of actions, fundamental differences arise. For instance, the uniqueness of the NE no longer holds. 

\noindent
Although the parallel MAC is, in terms of signal model, a special case of the MIMO (multiple input multiple output) interference channel (IC), the NE analysis conducted in  \cite{Scutari-Palomar-09, Berry-IT-2011, Larsson-Jorswieck-2009, Larsson-Jorswieck-2008, Yu-Cioffi-06} does not directly apply to the case of parallel MAC. First, \cite{Belmega-TWC-2009, ElGamal-08} address the case of fast fading links and shows substantial differences especially in terms of the uniqueness of the NE. This reveals the importance of the channel coherence time. Second, even if identical channel variation models are considered, the sufficient condition for the uniqueness of the NE derived in \cite{scutari-tsp_I-2008, Scutari-Palomar-09} for static or block fading IC, can be shown to hold with probability zero in parallel MAC \cite{Mertikopoulos-JSAC-2011}.

\noindent
Within this context, the main contributions of this paper are described hereunder.
\begin{itemize}
\item The set of NE of the decentralized parallel MAC is fully described in the case of CS policies and  single user decoding at the receiver. This set is shown to be non-empty and an upper bound of its cardinality is provided as a function of the number of transmitters and available channels.
\item In the $2$-transmitter $2$-channel case, it is formally proved that for any realization of the channel gains, there always exists at least one NE in the CS game that produces a higher or equal network spectral efficiency (NSE) than the unique NE in the PA game. In wireless communications, this kind of observations is often associated with a Braess type paradox \cite{Braess-69} as in \cite{Altman-wiopt-08b, Altman-wiopt-2008, Rose-Perlaza-2011}, where similar observations have been made in other scenarios. For an arbitrary number of transmitters and channels, we only provide numerical results that support the aforementioned claim.
\item The set of NE of the decentralized parallel MAC is also studied in the asymptotic regime, that is for a large number of transmitters in the case of CS policies. In this context, we provide closed-form expressions of the fraction of players which transmit over each channel as a function of the ratio between the channel bandwidth and the total available bandwidth (sum of all channel bandwidths).
\item Finally, we show show that in the case of successive interference cancellation (SIC), the power allocation vectors at the NE of the PA and CS games are capacity achieving. We provide both the power allocation vectors and the corresponding individual transmission rates at the NE such that the maximum NSE is achieved.
\end{itemize}

\noindent
The content of this paper can be briefly summarized as
follows. In Sec. \ref{SecModels}, the system and game models are
described. In Sec. \ref{SecPreliminaries}, we revisit the existing results regarding the PA problem and we provide new results on the CS problem in terms of existence and uniqueness of the NE. In Sec. \ref{SecSpecialCases}, the contribution aforementioned are fully detailed. In Sec. \ref{SecNumericalExamples}, we present simulation results in order to verify our theoretical results.
The paper is concluded by Sec. \ref{SecConclusions}.

\section{Models}\label{SecModels}

\subsection{System Model}\label{SecSignalModel}

Let us define the sets $\mathcal{K} \defn \lbrace 1, \ldots, K
\rbrace$ and $\mathcal{S} \defn \lbrace 1, \ldots, S \rbrace$.
Consider a parallel multiple access channel with $K$ transmitters
and $S$ subchannels (namely non-overlapping bands). Denote by
$\bs{y} = \left( y_1, \ldots, y_S\right)^T$ the $S$-dimensional
vector representing the received signal, which can be written in the
baseband at the symbol rate as follows
\begin{equation}
    \bs{y} = \displaystyle\sum_{k = 1}^K \bs{H}_k \bs{x}_k  + \bs{w}.
\end{equation}
Here,  $\forall k \in \mathcal{K}$,  $\bs{H}_k$ is the channel
 transfer matrix from transmitter $k$ to the receiver, $\bs{x}_k$ is
  the vector of symbols transmitted by transmitter $k$, and vector
  $\bs{w}$ represents the noise observed at the receiver. We will
   exclusively deal with the scenario where the matrices
    $\bs{H}_k$ are $S$-dimensional diagonal matrices (parallel MAC), i.e.,
     $\bs{H}_k = \diag\left( h_{k,1}, \ldots, h_{k,S} \right)$.
In our analysis, block fading channels are assumed. Hence, for each channel use, the entries $h_{k,s}$, for all $(k,s) \in \mathcal{K} \times \mathcal{S}$, are time-invariant realizations of a complex circularly symmetric Gaussian random variable with zero mean and unit variance. Here, we assume that each transmitter is able to perfectly estimate its own channel realizations (coherent communications), i.e., the channels $h_{k,1}\ldots h_{k,S}$.
The vector of transmitted symbols $\bs{x}_k$, $\forall k \in \mathcal{K}$, is
 an $S$-dimensional complex circularly symmetric Gaussian random variable
  with zero mean and covariance matrix $\bs{P}_k = \mathds{E}\left(\bs{x}_k \bs{x}_k^H\right)
   = \diag\left( p_{k,1} , \ldots, p_{k,S}\right)$. Assuming the transmit symbols to be Gaussian and independent is optimal
   in terms of spectral efficiency, as shown in \cite{Telatar-99, Telatar-95}. For all $(k,s) \in
\mathcal{K}\times\mathcal{S}$, $p_{k,s}$  represents the transmit
power allocated by transmitter $k$ over channel $s$ and transmitters
are power-limited, that is,
\begin{equation}
\forall k \in \mathcal{K}, \quad    \displaystyle\sum_{s = 1}^S p_{k,s} \leqslant p_{k,\max},
        \label{EqPowerConstraints}
\end{equation}
where $p_{k,\max}$ is the maximum transmit power of transmitter $k.$
A PA vector for transmitter $k$ is any
 vector $\bs{p}_k = \left(p_{k,1}, \ldots, p_{k,S}\right)$ with
  non-negative entries satisfying (\ref{EqPowerConstraints}).
The noise vector $\bs{w}$ is an $S$-dimensional zero mean Gaussian random
variable with independent, equal variance real and imaginary parts. Here,
$\mathds{E}\left(\bs{w} \bs{w}^H \right) = \diag\left(\sigma^2_1, \ldots, \sigma^2_S\right)$,
where, $\sigma^2_s$ represents the noise power over channel $s$.
We respectively denote the noise spectral density and the bandwidth of
channel $s$ by $N_0$ and $B_s$, thus, $\sigma^2_s = N_0 B_s$. The
total bandwidth is denoted by $B = \textstyle\sum_{s = 1}^S B_s$.

\subsection{Game Models}
The PA and CS problems described in Sec. \ref{SecIntroduction} can be
respectively modeled by the following two non-cooperative static
games in strategic form (with $i \in \{a,b\}$):
\begin{equation}
\mc{G}_{(i)} = \left(\mathcal{K},
\left(\mathcal{P}^{(i)}_k\right)_{k \in \mathcal{K}},
\left(u_k\right)_{k \in \mathcal{K}}\right).
\end{equation}
In both games, the set of transmitters $\mathcal{K}$ is the set of players. An action of a given transmitter $k \in \mathcal{K}$ is a particular PA scheme, i.e., an $S$-dimensional PA vector $\bs{p}_k = \left(p_{k,1}, \ldots, p_{k,S}\right) \in \mathcal{P}_k^{(i)}$, where $\mathcal{P}_k^{(i)}$ is the set of all possible PA vectors which transmitter $k$ can use either in the game $\gameone$ ($i = a$) or in the game $\gametwo$ ($i =b$). An action profile of the game $i \in \lbrace a, b \rbrace$ is a super vector
$$\bs{p} =\left(\bs{p}_{1}, \ldots, \bs{p}_{K}\right) \in \mathcal{P}^{(i)},$$
where $\mathcal{P}^{(i)}$ is a  set obtained from the Cartesian product of all the action sets, i.e., $\mathcal{P}^{(i)} = \mathcal{P}_1^{(i)} \times \ldots \times \mathcal{P}_{K}^{(i)} $, where,
\begin{eqnarray}
\nonumber
\mathcal{P}_k^{(a)} &=& \left\lbrace \left(p_{k,1}, \ldots p_{k,S}\right) \in \mathds{R}^{S}: \forall s \in \mathcal{S},  p_{k,s} \geqslant 0, \right. \\
\label{EqStrategySetGa} 
& & \left.
\displaystyle\sum_{ s \in \mathcal{S}} p_{k,s} \leqslant p_{k,\max} \right\rbrace, \mbox{ and }\\
\nonumber \mathcal{P}_k^{(b)} &=& \left\lbrace p_{k,\max} \, \bs{e}_{s}: \forall s \in \mathcal{S}, \; \bs{e}_s = \left(e_{s,1}, \ldots, e_{s,S}\right), \right. \\
\label{EqStrategySetGb} 
& & \left.
 \forall r \in \mathcal{S}\setminus s, e_{s,r} = 0, \text{ and } e_{s,s} = 1 \right\rbrace.
\end{eqnarray}
In the sequel, we respectively refer to the games $\gameone$ and
$\gametwo$ as the PA game and CS game. Let us denote by
$\bs{p}_{-k}$ any vector in the
 set 
 \begin{equation}\mathcal{P}_{-k}^{(i)} \defn \mathcal{P}_1^{(i)} \times \ldots
 \times \mathcal{P}_{k-1}^{(i)}\times \mathcal{P}_{k+1}^{(i)}\times
 \ldots \times \mathcal{P}_{K}^{(i)}
 \end{equation} 
 with $(i,k) \in \lbrace a, b \rbrace \times \mathcal{K}$. For a given $k \in \mathcal{K}$, the
  vector denoted by $\bs{p}_{-k}$ represents the strategies adopted by
   all the players other than player $k$. With a slight
    abuse of notation, we can write any vector
    $\bs{p} \in \mathcal{P}^{(i)}$
     as $\left(\bs{p}_k,\bs{p}_{-k}\right)$, in order to emphasize
      the $k$-th  component of the super vector $\bs{p}$. The utility
      for player $k$ in the game $\gamei$ is its spectral efficiency
      $u_k: \; \mathcal{P}^{(i)} \rightarrow \mathds{R}_+$, and
\begin{equation}\label{EqUtilityFunction}
u_k(\bs{p}_{k},\bs{p}_{-k}) = \displaystyle\sum_{s \in \mathcal{S}}
\frac{B_s}{B} \log_2\left(1 + \gamma_{k,s}\right) \mbox{   [bps/Hz]},
\end{equation}
where $\gamma_{k,s}$ is the signal-to-interference
 plus noise ratio (SINR) seen by player $k$ over its channel $s$, i.e.,
\begin{equation}\label{EqSINR}
    \gamma_{k,s} = \frac{p_{k,s} g_{k,s}}{\sigma^2_s +
    \ds\sum_{j \in \mathcal{K}\setminus\lbrace k \rbrace}p_{j,s} g_{j,s}},
\end{equation}
and $g_{k,s} \triangleq \left| h_{k,s}\right|^2$. Note that from \eqref{EqSINR}, it is implied that single-user decoding (SUD) is used at the
receiver. 
The choice of SUD is basically due to scalability (in terms of signaling cost) and fairness for the
decoding scheme. 
Clearly, optimality is not sought here, nonetheless, these constraints are inherit to the decentralized nature of the network. 

\noindent
As a solution concept for both $\gameone$ and $\gametwo$, we focus on the notion of NE \cite{Nash-1950}, which we define, using our notation, as follows,
\begin{definition}[Pure Nash Equilibrium]\label{DefNE} \emph{In the
non-cooperative games in strategic form $\mathcal{G}_{(i)}$, with $i
\in \lbrace a, b \rbrace$, an action profile $\boldsymbol{p} \in
\mathcal{P}^{(i)}$ is a pure NE if it satisfies, for all $k \in
\mathcal{K}$ and for all $\bs{p}'_k \in \mathcal{P}_k^{(i)}$, that
\begin{equation}
u_k(\bs{p}_k,\bs{p}_{-k}) \geqslant  u_k(\bs{p}'_k,\bs{p}_{-k}).
\end{equation}
}
\end{definition}
The relevance of the NE is that at such state, the PA or CS policy chosen by any transmitter is optimal with respect to the choice of all the other transmitters. Thus, in a decentralized network, the NE is a stable network state, since no player has a particular interest in deviating from the actual state.

\noindent
In the following, we provide some fundamental results which we use in the further analysis of the games $\gameone$ and $\gametwo$.
%
\section{Auxiliary Results}
\label{SecPreliminaries}

In this section, we introduce some existing results on the existence and uniqueness of the NE in the games $\gameone$ and $\gametwo$. For doing so, we use the fact that both games $\gameone$ and $\gametwo$ have been shown to be potential games (PG) \cite{Monderer-Shapley-1996, Scutari-2006} in \cite{Perlaza-Gamecomm-09} and \cite{Perlaza-Latincom-2010}, respectively. We conclude this section, by introducing a new result which allows us to establish an upper bound on the number of NE that the game $\gametwo$ might possess.

\subsection{General Results}

The analysis presented in the following strongly relies on the fact that both games $\gameone$ and $\gametwo$ are potential games. Thus, for the sake of completeness, we define exact PG using our notation.
\begin{definition}[Exact Potential Game]\label{DefPotentialGame} \emph{Any game in strategic form defined by the triplet
$$\mathcal{G} = \left( \mathcal{K},\left(\mathcal{P}_k\right)_{k \in \mathcal{K}}, \left(u_k\right)_{k \in \mathcal{K}}\right)$$
is an exact potential game if there exists a function $\phi\left(\bs{p}\right)$ for all $\bs{p} \in \mathcal{P} = \mathcal{P}_{1}\times\ldots\times\mathcal{P}_{K}$ such that
 for all players $k \in \mathcal{K}$ and for all $\bs{p}'_k \in \mathcal{P}_k$, it holds that
\begin{equation}\nonumber
 u_k(\bs{p}_{k},\bs{p}_{-k}) - u_k(\bs{p}'_{k},\bs{p}_{-k})
 = \phi(\bs{p}_{k},\bs{p}_{-k}) - \phi(\bs{p}'_{k},\bs{p}_{-k}).
\end{equation}
}
\end{definition}
Now, from the definition of the utility functions (\ref{EqUtilityFunction}) and Def. \ref{DefPotentialGame}, the results proved by the authors in \cite{Perlaza-Gamecomm-09} and \cite{Perlaza-Latincom-2010} follow immediately.
\begin{lemma}\label{LemmaGisPG}
\emph{The strategic form games $\mathcal{G}_{(i)}$, with $i \in \lbrace a, b\rbrace$, are exact potential games with potential function
\begin{equation}\label{EqPotential}
     \phi(\bs{p})= \displaystyle\sum_{s \in \mathcal{S}}\frac{B_s}{B}\log_2\left(
     \sigma^2_s  +  \displaystyle\sum_{k = 1}^K p_{k,s} g_{k,s}\right).
\end{equation}
}
\end{lemma}
The relevance of PG relies on the fact that it is a class of games for which the existence of at least one pure NE is always guaranteed \cite{Monderer-Shapley-1996}.
Additionally, many known learning procedures, such as, best response dynamics, fictitious play and some reinforcement learning dynamics converge in PG. As a consequence, any of these dynamics can be used to implement algorithms to achieve an equilibrium in a fully decentralized fashion. Nonetheless, we leave the design of decentralized techniques for achieving NE out of the scope of this paper and we focus on the analysis of the equilibria. We refer the interested reader to \cite{Lasaulce-Tembine-Book-2011, scutari-tsp_II-2008, Scutari-Palomar-09, Yu-2004, Rose-Perlaza-CommMag-2011} for more details.
In the following, we use Prop. \ref{LemmaGisPG} to analyze  the set of NE of both our power allocation game $\gameone$ and our channel selection game $\gametwo$.

\subsection{Known Results Concerning the Power Allocation Game $\gameone$}\label{SecGa}

In the following, we comment on the existence and uniqueness of the NE in $\gameone$.

\subsubsection{Existence of a NE}

Regarding the existence of pure NE, the following proposition is an immediate consequence of our Lemma \ref{LemmaGisPG} and \emph{Lemma $4.3$} in \cite{Monderer-Shapley-1996}.
\begin{proposition}[Existence of a pure NE]\label{PropExistenceGa} \emph{The
game  $\gameone$ has always at least one NE in pure strategies.}
\end{proposition}

\noindent
Regarding the existence of a mixed NE (i.e., a probability distribution on the possible actions which verifies Definition \ref{DefNE}), it follows from \cite{Fan-1952}, that the existence of at least one NE in mixed strategies always exists. This is basically
because the action spaces, $\mathcal{P}_k^{(a)}$  are compact spaces and the utility functions are
continuous with respect to the action profile. However, in compact strategy spaces,
mixed strategies are generally less attractive due to the difficulty of its implementation in wireless communications systems \cite{Lasaulce-Tembine-Book-2011}.

\subsubsection{Uniqueness of the pure NE}
In the game $\gametwo$, the uniqueness of the NE has been shown to hold with probability one \cite{Mertikopoulos-JSAC-2011}.
\begin{theorem}[NE uniqueness in parallel MAC]\label{theo:uniqueness-KS} The game
$\mc{G}_{(a)}$ has almost surely a unique pure NE.
\end{theorem}
A formal proof of Theorem \ref{theo:uniqueness-KS} is provided in \cite{Mertikopoulos-JSAC-2011}.  This proof is based on the concept of degeneracy which allows one to characterize the directions along which the potential remains constant. A simpler proof, for the special case of $2$
-players and $2$-channels is given in \cite{Perlaza-PhdThesis}.
\subsubsection{Determination of the NE}

From Def. \ref{DefPotentialGame}, it follows that the unique NE in pure strategies, denoted by $\bs{p}^{\dagger}$, is the unique solution of the following optimization problem:
\begin{equation}\label{EqOPGameA}
    \displaystyle\arg\max_{p \in \mathcal{P}^{(a)}} \phi\left( \bs{p}
    \right).
\end{equation}
The components of the PA vector $\bs{p}^{\dagger}$ in (\ref{EqOPGameA})
are for all $(k,s) \in \mathcal{K}\times\mathcal{S}$,

\begin{equation}\label{EqNEGameA}
\begin{array}{lcl}
p_{k,s}^{\dagger} & = &         \left[\displaystyle \frac{B_s}{B}
\frac{1}{\beta_k} - \frac{\sigma_s^2 + \ds\sum_{j \in \mathcal{K}\setminus
                                                                                                \lbrace k \rbrace} p_{j,s}^{\dagger} g_{j,s}}{\textstyle g_{k,s}} \right]^+,
\end{array}
\end{equation}
where, $\beta_k$ is a Lagrangian multiplier  chosen to saturate the power constraints (\ref{EqPowerConstraints}). Note that this result shows the connections between the notion of NE and the well-known water-filling PA policy \cite{Yu-2004}. For a further discussion on this connection, the reader is referred to  \cite{Lasaulce-Tembine-Book-2011, Rose-Perlaza-CommMag-2011}.

\subsection{New Results Concerning the Channel Selection Game $\gametwo$}\label{SecGb}

In the game $\gametwo$, it can be checked that given a vector $\bs{p}_{-k} \in \mathcal{P}^{(b)}_{-k}$, it follows that  $\forall k \in \mathcal{K}$ 
and $\forall p_k \in \left[ 0, p_{k,\max}\right]$, the potential function satisfies that $\phi\left(p_k \;
\bs{e}_{s},\bs{p}_{-k}\right) \leqslant 
\phi\left(p_{k,\max}\;\bs{e}_{s},\bs{p}_{-k}\right)$, where $\bs{e}_s = \left(e_{s,1}, \ldots, e_{s,S}\right) \in \mathds{R}^S$, $\forall r \in \mathcal{S}\setminus s$, $e_{s,r} = 0$, and $e_{s,s} = 1$. Thus, the problem of transmit power control disappears and there is no lost of generality by choosing the action sets as $\mathcal{P}_{k}^{(b)}$. Technically, the main difference between $\gameone$ and $\gametwo$ is that the latter is a finite game ($|\mc{K}\times\mc{S}|<+\infty$). In the following, we investigate the consequences of this fact on the existence and multiplicity of the NE .

\subsubsection{Existence of a pure NE}

Regarding the existence of at least one NE in pure strategies, we have that from our Prop. \ref{LemmaGisPG} and \emph{Corollary $2.2$} in \cite{Monderer-Shapley-1996}, the following proposition holds.
\begin{proposition}[Existence of a pure NE]\label{PropExistenceGb} \emph{The game
  $\gametwo$ has always at least one NE in pure strategies.}
\end{proposition}
Regarding the existence of an equilibrium in mixed strategies, we have that given that the actions sets are discrete and
finite, then the existence of at least one NE in mixed strategies is ensured \cite{Nash-1950}.
%
\subsubsection{Multiplicity of the pure NE}
In the following, we introduce a new result that allows us to determine  that the NE in the game $\gametwo$ is not necessarily unique.
\begin{proposition}\label{PropMultiplicityGb} \emph{Let $\hat{K}\in\mathds{N}$ be the highest even number which is less or equal to $K$. Then, the game
$\gametwo$ has   $L $ pure NE strategy profiles, where,
\begin{equation}\label{EqL}
1 \leqslant L
\leqslant 1 + \left(S - 1\right)\ds\sum_{i \in \lbrace 2, 4, \ldots, \hat{K}\rbrace} \left(\begin{array}{c} K \\ i\end{array}\right).
\end{equation}
}
\end{proposition}
The proof of Proposition \ref{PropMultiplicityGb} is given in Appendix \ref{AppProofOfPropMultiplicityGb}.
An interesting point in Proposition \ref{PropMultiplicityGb} is that it shows that the property of uniqueness in the game $\gameone$ no longer holds (Theorem \ref{theo:uniqueness-KS} ). By quantizing the set of actions, the game loses the uniqueness property for the NE.
The upper bound of the number of NE in Proposition \ref{PropMultiplicityGb} does not depend on the realizations of the channel
gains but only on the number of transmitters $K$ and available channels $S$. As we shall see, when the exact realizations of the channels are known, then the exact number of NE can be identified. Nonetheless, relying only on the number of transmitters $K$ and number of channels $S$, the tightest upper bound is given by Proposition \ref{PropMultiplicityGb}. 

\subsubsection{Determination of the NE}\label{SecDeterminationOfNEGameb}

In order to fully identify the action profiles corresponding to an NE, we use the graph  $G$ described in the proof of Proposition \ref{PropMultiplicityGb}  in Appendix \ref{AppProofOfPropMultiplicityGb}. Basically, we convert the non-directed graph $G$ into an oriented graph
$\hat{G}$ whose adjacency matrix is the  non-symmetric square matrix
$\bs{\hat{A}}$ whose entries are $\forall (i,j) \in \mathcal{I}^2$
and  $i\neq j$,
\begin{equation}\label{EqMatrixAHat}
    \hat{a}_{ij} = \left\lbrace \begin{array}{ccl} 1 & \text{ if }        & i \in \mathcal{V}_j \text{ and } \phi\left(\bs{p}^{(j)}\right) > \phi\left(\bs{p}^{(i)}\right)\\
                                                                                                    0 & \text{ otherwise }, &
                                                                     \end{array} \right.
\end{equation}
and $\hat{a}_{i,i} = 0$ for all $i \in \mathcal{I}$.

\noindent
From the definition of the matrix $\hat{\bs{A}}$, we have that a
necessary and sufficient condition for a vertex $v_i$ to represent
an NE action profile is to have a null out-degree in the oriented
graph $\hat{G}$, i.e., there are no outgoing edges from the node
$v_i$ (sink vertex).
Finally, one can conclude that determining the set of NE in the game $\gametwo$ boils down to
identifying all the sink vertices in the oriented graph $\hat{G}$.
In Fig. \ref{FigGraph}, we show an example of the non-directed $G$ and oriented $\hat{G}$ graphs for the case where $K = 3$ and $S = 2$.
Note that this method can be used only to determine the whole set of NE. It does not pretend to be an algorithm which can be directly implemented in decentralized wireless networks, since it requires complete information at each transmitter. 
Methods for achieving equilibria in wireless networks are described in \cite{Lasaulce-Tembine-Book-2011, scutari-tsp_II-2008, Scutari-Palomar-09, Yu-2004, Rose-Perlaza-CommMag-2011}.

\section{Equilibrium Performance Analysis and Special Cases}\label{SecSpecialCases}
In this section, we study in detail two special cases of relevant interest to understand previous conclusions and provide more insights into decentralized power allocation problem in terms of design.
First, the games $\gameone$ and $\gametwo$ are studied assuming that there exist only $K=2$ transmitters and $S=2$ available channels.
In particular, we analyze the set of NE action profiles of both games and compare the network spectral efficiency (NSE), $U^{(i)}: \mathcal{P}^{(i)} \rightarrow \mathds{R}$, obtained by playing both games. Here, for all $i \in \lbrace 1, 2 \rbrace$,
\begin{equation}\label{EqGlobalUtilityFunction}
U^{(i)}(\bs{p}_1, \ldots,\bs{p}_{K}) = \displaystyle\sum_{k=1}^K u_k(\bs{p}_1, \ldots,\bs{p}_{K}) \mbox{  [bps/Hz]}.
\end{equation}
From this analysis, we conclude that from a network performance point of view, limiting the transmitters to use a unique channel brings a better result in terms of network spectral efficiency \eqref{EqGlobalUtilityFunction}.
Second, we consider the case of a large number of transmitters. This study leads to two important conclusions: $(i)$ the fraction of players using a given channel depends mainly on the bandwidth of each channel and not on the exact channel realization nor the number of players and channels; $(ii)$ in the asymptotic regime ($K \rightarrow \infty$) both games exhibit the same performance.
Before we start, let us introduce the notion of best response correspondence, since it plays a central role in the following analysis.
\begin{definition}[Best-Response Correspondence]\label{DefBR} \emph{In a non-cooperative game described by the $3$-tuple $\left(\mathcal{K},\left(\mathcal{P}_k\right)_{\forall k \in \mathcal{K}},\left(u_{k}\right)_{\forall k \in \mathcal{K}}\right)$, the relation $\BR_k: \mathcal{P}_{-k} \rightarrow \mathcal{P}_{k}$ such that
    \begin{equation}\label{EqBRi}
         \BR_{k}\left(\bs{p}_{-k}\right) =
         \ds\arg\max_{\bs{q}_k \in \mathcal{P}_k}
            u_k\left(\bs{q}_k,\bs{p}_{-k}\right),
    \end{equation}
is defined as the best-response correspondence of player $k \in
\mathcal{K}$, given the actions $\bs{p}_{-k}$ adopted by all the
other players. }
\end{definition}

\subsection{The $2$-Transmitter $2$-Channel Case}\label{Sec2x2Game}

Consider the games $\gameone$ and $\gametwo$ with $K = 2$ and $S = 2$. Assume also that $\forall k \in \mathcal{K}$, $p_{k,\max} = p_{\max}$ and $\forall s \in \mathcal{S}$, $\sigma^2_s = \sigma^2$ and $B_s = \frac{B}{S}$. Denote by $\SNR = \frac{p_{\max}}{\sigma^2}$ the average signal to noise ratio (SNR) of each active communication.

\subsubsection{The power allocation game}

Let us denote by $\bs{p}^{\dagger} = \left(\bs{p}_1^{\dagger},\bs{p}_2^{\dagger}\right)$ the NE of the game $\gameone$. Then, following Def. \ref{DefNE}, one can write the following set of inclusions,
\begin{equation}\label{EqSystemNE2x2}
\forall k \in \mathcal{K}, \quad    \bs{p}_k^{\dagger} \in \BR_{k}\left(\bs{p}_{-k}^{\dagger}\right).
\end{equation}
Note that, for all $k \in \mathcal{K}$ and for all $\bs{p}_{-k} \in \mathcal{P}^{(a)}$, the set $\BR_{k}\left(\bs{p}_{-k}\right)$ is a singleton (Def. \ref{DefBR}) and thus, \eqref{EqSystemNE2x2} represents a system of equations.
By solving the resulting system of equations (\ref{EqSystemNE2x2}) for a given realization of the channels $\channelset$, one can determine the NE of the game $\gameone$. We present such a solution in the following proposition.
\begin{proposition}[Nash Equilibrium in $\gameone$]\label{PropNEGa} \emph{Let the action profile $\bs{p}^{\dagger} = \left(\bs{p}_1^{\dagger},\bs{p}_2^{\dagger}\right) \in \mathcal{P}^{(a)}$, with $\bs{p}^{\dagger}_{1}=\left(p_{11}^{\dagger},p_{\max}-p_{11}^{\dagger}\right)$ and  $\bs{p}^{\dagger}_{2}=\left(p_{\max}-p_{22}^{\dagger},p_{22}^{\dagger}\right)$ be an NE action profile of the game $\gameone$. Then, with probability one, $\bs{p}^{\dagger}$ is the unique NE and it can be written as follows:
\begin{itemize}
    \item Equilibrium $1$: if $\bs{g} \in \mathcal{B}_1 = \lbrace \bs{g} \in \mathds{R}_+^4: \frac{g_{11}}{g_{12}} \geqslant \frac{1 +\SNR g_{11}}{1 + \SNR g_{22}}, \, \frac{g_{21}}{g_{22}} \leqslant \frac{1 +\SNR g_{11}}{1 + \SNR g_{22}}\rbrace,
$ then, $p_{11}^{\dagger} =  p_{\max}$ and $p_{22}^{\dagger} =  p_{\max}$.
    \item Equilibrium $2$: if $\bs{g} \in \mathcal{B}_2 = \lbrace \bs{g} \in \mathds{R}_+^4: \frac{g_{11}}{g_{12}} \geqslant 1 + \SNR \left(g_{11} + g_{21}\right), \, \frac{g_{21}}{g_{22}} \geqslant 1 + \SNR \left(g_{11} + g_{21}\right)\rbrace$,
then, $p_{11}^{\dagger} = p_{\max}$ and $p_{22}^{\dagger} = 0$.%
    \item Equilibrium $3$: if $\bs{g} \in \mathcal{B}_3 = \lbrace \bs{g} \in \mathds{R}_+^4: \frac{g_{11}}{g_{12}} \leqslant \frac{1}{1 + \SNR \left(g_{12} + g_{22}\right)}, \, \frac{g_{21}}{g_{22}} \leqslant \frac{1}{1 + \SNR\left(g_{12} + g_{22}\right)}\rbrace$ then, $p_{11}^{\dagger} = 0$ and $p_{22}^{\dagger} = p_{\max}$.
    \item Equilibrium $4$: if $\bs{g} \in \mathcal{B}_4 = \lbrace \bs{g} \in \mathds{R}_+^4:  \frac{g_{11}}{g_{12}} \leqslant \frac{1 + \SNR g_{21}}{1 + \SNR g_{12}}, \, \frac{g_{21}}{g_{22}} \geqslant \frac{1 + \SNR g_{21}}{1 + \SNR g_{12}}\rbrace$,
then, $p_{11}^{\dagger} = 0$ and $p_{22}^{\dagger} =  0$.
    \item Equilibrium $5$: if $\bs{g} \in \mathcal{B}_{5} \lbrace \bs{g} \in \mathds{R}_+^4: \frac{g_{11}}{g_{12}}  \geqslant  \frac{g_{21}}{g_{22}}, \, \frac{1 + \SNR g_{11}}{1+\SNR g_{22}} < \frac{g_{21}}{g_{22}} <  1 + \SNR  \left(g_{11}+ g_{21}\right)\rbrace$, then, $p_{11}^{\dagger} = p_{\max}$ and  $p_{22}^{\dagger} = \frac{1}{2}\left(p_{\max} - \frac{\sigma^2}{g_{22}} + \frac{\sigma^2+g_{11} p_{\max}}{g_{21}}\right)$.
    \item Equilibrium $6$: if $\bs{g} \in \mathcal{B}_6 \lbrace \bs{g} \in \mathds{R}_+^4: \frac{g_{11}}{g_{12}}  \geqslant  \frac{g_{21}}{g_{22}}, \, \frac{1}{1 + \SNR \left(g_{12} + g_{22}\right)} < \frac{g_{11}}{g_{12}} < \frac{1 + \SNR  g_{11}}{1 + \SNR g_{22}}\rbrace$, then, $p_{11}^{\dagger} = \frac{1}{2}\left(p_{\max} - \frac{\sigma^2}{g_{11}} + \frac{\sigma^2 + p_{\max}g_{22}}{g_{12}}\right)$ and $ p_{22}^{\dagger} =  p_{\max}$.
    \item Equilibrium $7$: if $\bs{g} \in \mathcal{B}_{7} = \lbrace \bs{g} \in \mathds{R}_+^4: \frac{g_{11}}{g_{12}} \leqslant \frac{g_{21}}{g_{22}}, \, \frac{1 + \SNR  g_{21}}{1 + \SNR  g_{12}} < \frac{g_{11}}{g_{12}} < 1 + \SNR\left(g_{11} + g_{21} \right)\rbrace$, then, $p_{11}^{\dagger} = \frac{1}{2}\left(p_{\max} - \frac{\sigma^2 + p_{\max}g_{21}}{g_{11}} + \frac{\sigma^2}{g_{12}}\right)$ and $p_{22}^{\dagger}  = 0$.
\item Equilibrium $8$: if $\bs{g} \in \mathcal{B}_{8} \lbrace \bs{g} \in \mathds{R}_+^4: \frac{g_{11}}{g_{12}}  \leqslant \frac{g_{21}}{g_{22}}, \,  \frac{1}{1 + \SNR \left(g_{12} +g_{22} \right)} < \frac{g_{21}}{g_{22}} <  \frac{1 + \SNR g_{21}}{1 + \SNR g_{12}}\rbrace$, then, $p_{11}^{\dagger} = 0$ and $p_{22}^{\dagger} = \frac{1}{2}\left(p_{\max} - \frac{\sigma^2 + g_{12} p_{\max}}{g_{22}} + \frac{\sigma^2}{g_{21}}\right)$.
    \end{itemize}
}
    \end{proposition}
\begin{proof} See Appendix \ref{AppProf1}
\end{proof}

\noindent
In Fig. \ref{FigNashRegionsG1} we plot the different types
of NE of the game $\gameone$ as a function of the channel ratios
$\frac{g_{11}}{g_{12}}$ and $\frac{g_{21}}{g_{22}}$. Note that under the knowledge of all channels, the set of all possible pure NE can be obtained by simply placing the point $\left( \frac{g_{11}}{g_{12}}, \frac{g_{21}}{g_{22}}\right)$ in Fig. \ref{FigNashRegionsG1} .
The uniqueness of the NE is not ensured under certain conditions. In fact, infinitely many NE can be observed, however, the conditions for this to happen are zero probability events, as we shall see.

\begin{proposition}\label{LemmaNonUniqueness} \emph{Let $\alpha \defn \frac{g_{11}}{g_{21}} = \frac{g_{12}}{g_{22}}$ and assume that the set of channels $\channelset$ satisfies the following conditions
\begin{eqnarray}\nonumber
\frac{1}{1 + \frac{p_{\max}}{\sigma^2}(g_{12}+g_{22})} < & \alpha & <  1 + \frac{p_{\max}}{\sigma^2}\left( g_{11} + g_{21} \right),
\end{eqnarray}
Then, any PA vector $\bs{p} = \left(p_{11},p_{\max} - p_{11}, p_{\max} - p_{22}, p_{22}\right) \in \mathcal{P}^{(a)}$, such that 
$$p_{11} = \textstyle\frac{1}{2}\left(p_{\max} \left(1 - \alpha \right) + \sigma^2 \left(\frac{1}{g_{12}} - \frac{1}{g_{11}}\right) \right) + \alpha p_{22}$$ 
is an NE action profile of the game $\gameone$.
}
\end{proposition}
The proof of Prop. \ref{LemmaNonUniqueness} is the first part of the proof of Prop. \ref{PropNEGa}.
In the next subsection, we perform the same analysis presented above for the game $\gametwo$.
\subsubsection{The channel selection game}\label{SecNEinDiscrete}
When $K = 2$ and $S = 2$, the game $\gametwo$ has four possible
outcomes, i.e., $\left|\mathcal{P}^{(b)}\right| = 4$. We detail such
outcomes and its respective potential in Fig. \ref{TabPotential}.
Following Def. \ref{DefNE}, each of those outcomes can be an NE
depending on the channel realizations $\channelset$, as shown in the
following proposition.
\begin{figure}[h]
\begin{center}
$\begin{array}{|c|c|c|}\hline
 \scriptscriptstyle Tx_1 \backslash Tx_2 
 &   \scriptscriptstyle \bs{p}_2 = \left(p_{\max},0\right) &   \scriptscriptstyle \bs{p}_2 = \left(0, p_{\max}\right) \\ \hline
\begin{array}{c} \scriptscriptstyle  \bs{p}_1 = \\ \scriptscriptstyle   \left(p_{\max},0\right) \end{array} & \begin{array}{l}\scriptscriptstyle   \frac{1}{2} \log_2\left(\sigma^2 + p_{\max}(g_{11} + g_{21})\right) \\ \scriptscriptstyle + \frac{1}{2} \log_{2}\left(\sigma^2 \right) \end{array} &  \begin{array}{l}  \scriptscriptstyle \frac{1}{2} \log_2\left(\sigma^2 + p_{\max} g_{11}\right) \\ \scriptscriptstyle + \frac{1}{2} \log_{2}\left(\sigma^2 + p_{\max} g_{22}\right) \end{array} \\ \hline
\begin{array}{c} \scriptscriptstyle  \bs{p}_1 = \\ \scriptscriptstyle \left(0, p_{\max}\right) \end{array}&
 \begin{array}{l}  \scriptscriptstyle \frac{1}{2} \log_2\left(\sigma^2 + p_{\max} g_{12}\right) \\ \scriptscriptstyle  + \frac{1}{2} \log_{2}\left(\sigma^2 + p_{\max} g_{21}\right) \end{array}
 &  \begin{array}{l}  \scriptscriptstyle \frac{1}{2} \log_2\left(\sigma^2 + p_{\max} (g_{12} + g_{22})\right)  \\  \scriptscriptstyle + \frac{1}{2} \log_{2}\left(\sigma^2 \right)\end{array} \\ \hline
\end{array}$
\end{center}
\caption{Potential function $\phi$ of the game $\gameone$, with $K = 2$ and $S = 2$. Player $1$ chooses rows and player $2$ chooses columns.}
\label{TabPotential}
\end{figure}
\begin{proposition}[Nash Equilibria in $\gametwo$]\label{PropNEGb} \emph{Let the PA vector $\bs{p}^* = \left(\bs{p}_{1}^*,\bs{p}_{2}^*\right) \in \mathcal{P}^{(b)}$ be one NE in the game $\gametwo$. Then, depending on the channel gains $\channelset$, the NE $\bs{p}^*$ can be written as follows:
\noindent
\begin{itemize}
\item Equilibrium $1$: when $\bs{g} \in \mathcal{A}_{1}  = \lbrace \bs{g} \in \mathds{R}_+^4: \frac{g_{11}}{g_{12}} \geqslant \frac{1}{1 + \SNR g_{22}} \mbox{ and } \frac{g_{21}}{g_{22}} \leqslant 1 + \SNR  g_{11} \rbrace$, then, $\bs{p}_1^* = \left(p_{\max},0\right)$ and $\bs{p}_2^* = \left(0,p_{\max}\right)$.
\item Equilibrium $2$: When $\bs{g} \in \mathcal{A}_{2} = \lbrace \bs{g} \in \mathds{R}_+^4:  \frac{g_{11}}{g_{12}} \geqslant 1 + \SNR g_{21} \mbox{ and } \frac{g_{21}}{g_{22}} \geqslant 1 + \SNR g_{11} \; \rbrace$, then, $\bs{p}_1^* = \left(p_{\max},0\right)$ and $\bs{p}_2^* = \left(p_{\max},0\right)$.
\item Equilibrium $3$: when $\bs{g} = \left(g_{11},g_{12},g_{21},g_{22}\right) \in \mathcal{A}_{3} \lbrace \bs{g} \in \mathds{R}_+^4: \frac{g_{11}}{g_{12}} \leqslant \frac{1}{1 + \SNR g_{22}}  \mbox{ and } \frac{g_{21}}{g_{22}} \leqslant \frac{1}{1 + \SNR g_{12}} \rbrace$, then, $\bs{p}_1^* = \left(0,p_{\max}\right)$ and $\bs{p}_2^* = \left(0,p_{\max}\right)$.
\item Equilibrium $4$: when $\bs{g} \in \mathcal{A}_{4} = \lbrace \bs{g} \in \mathds{R}_+^4: \frac{g_{11}}{g_{12}} \leqslant 1 + \SNR g_{12} \mbox{ and }  \frac{g_{21}}{g_{22}} \geqslant \frac{1}{1 + \SNR g_{12}} \; \rbrace$, then,  $\bs{p}_1^* = \left(0,p_{\max}\right)$ and $\bs{p}_2^* = \left(p_{\max},0\right)$.
\end{itemize}
}
\end{proposition}
\begin{proof} The proof follows immediately from Def. \ref{DefNE} and Tab. \ref{TabPotential}.
\end{proof}
\noindent
In Fig. \ref{FigNashRegionsG2}, we plot the different types of NE action profiles as a function of the channel ratios $\frac{g_{11}}{g_{12}}$ and $\frac{g_{21}}{g_{22}}$. Note that under the knowledge of all channels, the set of all possible pure NE can be obtained by simply placing the point $\left( \frac{g_{11}}{g_{12}}, \frac{g_{21}}{g_{22}}\right)$ in \ref{FigNashRegionsG2}.
Note how the action profiles $\bs{p}^{*} = \left(p_{\max},0,0,p_{\max}\right)$ and $\bs{p}^{+} = \left(0,p_{\max},p_{\max},0\right)$ are both NE, when the channel realizations satisfy that $\bs{g} \in \mathcal{A}_5 = \mathcal{A}_{1} \cap \mathcal{A}_{4}$, i.e.,
$\mathcal{A}_5 =  \lbrace \bs{g} \in \mathds{R}_+^4: \frac{1}{1+ \SNR g_{22}} \leqslant \frac{g_{11}}{g_{12}} \leqslant 1+ \SNR g_{21} \text{ and } \frac{1}{1+ \SNR g_{12}} \leqslant \frac{g_{21}}{g_{22}} \leqslant  1+ \SNR g_{11} \rbrace$.
\noindent This confirms the fact that several NE might exists in the
game $\gametwo$ depending on the exact channel realization, as
stated in Prop. \ref{PropMultiplicityGb}. Moreover, one can also
observe that there might exist an NE action profile which is not a global
maximizer of the potential function (\ref{EqPotential}) \cite{Ui-08}
(e.g., $\phi\left(\bs{p}^*\right) <
\phi\left(\bs{p}_{2}^{+}\right)$).

\noindent
In the sequel, the performance
achieved by the transmitters at the equilibrium in both games are
compared.
\subsection{A Braess Type Paradox}
As suggested in \cite{Braess-69}, a Braess-type paradox refers to a counter-intuitive observation consisting in a reduction of the individual utility at the equilibria, when the players are granted with a larger set of actions. That is, by letting the players to choose from a larger set of options, their individual benefit reduces. Recently, the Braess-type paradox has been also associated with the reduction of the sum-utility instead of the individual utilities, as in \cite{Altman-wiopt-2008, Altman-03,Rose-Perlaza-2011} in the wireless communications arena.

\noindent
In our particular case, the set of actions for player $k$ in the game $\gametwo$, is a
subset of its set of actions in the game $\gameone$, i.e., $\forall
k \in \mathcal{K}$, $\mathcal{P}_{k}^{(b)} \subseteq
\mathcal{P}^{(a)}_{k}$. Interestingly, as observed in \cite{Braess-69}, reducing the set of actions of each player leads, in this particular game, to a better global performance. This effect has been reported in  the parallel interference channel in \cite{Altman-wiopt-2008, Altman-wiopt-08b} under the consideration of particular channel conditions and later, more generally in \cite{Rose-Perlaza-2011}. However, a formal proof of the existence of this paradox is not provided in the aforementioned references.
This observation, has been also reported in the parallel MAC for the case of successive interference cancellation (SIC) in
\cite{Perlaza-Crowncom-09}. Nonetheless, the channel in \cite{Perlaza-Crowncom-09} was not fully decentralized, as it required a central controller to dictate the channel policy to all the transmitters.
In the following, we study this observation in more detail.

\noindent
Let us denote by $\bs{p}_k^{(\dagger,n)}$, the unique NE action profile of game $\gameone$, when the vector $\bs{g} = \channelvector \in \mathcal{B}_{n}$, for all $n \in \lbrace 1 \ldots, 8 \rbrace$. Let us also denote by $\bs{p}^{(*,n)}$ one of the NE action profiles of game $\gametwo$ when $\channelvector \in \mathcal{A}_n$, for all $n \in \lbrace 1, \ldots 4\rbrace$.
The sets $\mathcal{A}_n$ and $\mathcal{B}_n$ are defined in Prop. \ref{PropNEGa} and \ref{PropNEGb}.
Then, for a finite SNR level, one can observe that $\forall n \in \lbrace 1, \dots, 4\rbrace$, $\mathcal{A}_{n} \cap \mathcal{B}_{n} = \mathcal{B}_{n}$ and $\forall \bs{g} = \channelvector \in \mathcal{B}_n$, the following equality always holds $\bs{p}_k^{(\dagger,n)} = \bs{p}_k^{(*,n)}$, which implies the same network performance. However, when the NE of both games are different, one can not easily compare the utilities achieved by each player since they depend on the exact channel realizations. Fortunately, the analysis largely simplifies by considering either a low SNR regime or a high SNR regime and more general conclusions can be stated. The performance comparison between games $\gameone$ and $\gametwo$ for the low SNR regime is presented in the following proposition.

\begin{proposition}\label{PropPerformanceLowSNR}
\emph{In the low SNR regime, both games $\gameone$ and $\gametwo$, with $K = 2$ and $S=2$, possess a unique NE, denoted by $\bs{p}^*$. Here, for all $k \in \mathcal{K}$ and $s \in \mathcal{S}$,
\begin{eqnarray}
    p_{k,s}^*  & = & p_{\max} \mathds{1}_{\left\lbrace s = \ds\arg\max_{\ell \in \mathcal{S}} g_{k,\ell}\right\rbrace}\label{EqPAinLowSNR1}\\
    p_{k,-s}^* & = & p_{\max} - p_{k,s}\label{EqPAinLowSNR2}.
\end{eqnarray}
}
\end{proposition}
\begin{proof}See App. \ref{ProofPerformanceLowSNR} \end{proof}
From, Prop \ref{PropPerformanceLowSNR}, it can be stated that at the low SINR regime, players achieve the NE by simply choosing the channel with the highest channel gain independently of the other player's action. 
The performance comparison between games $\gameone$ and $\gametwo$ for the high SNR regime is presented in the following proposition.
\begin{proposition}\label{PropPerformanceHighSNR} \emph{In the high SNR regime, with $K = 2$ and $S=2$, the game $\gameone$ has a unique pure NE denoted by $\bs{p}^{\dagger}$ and the game $\gametwo$ has two pure NE denoted by $\bs{p}^{(*,1)}$ and $\bs{p}^{(*,4)}$, respectively. Then, at least for one $n \in \lbrace 1, 4\rbrace$, there exists a SNR value $0 < \SNR_0 < \infty$,  such that $\forall \SNR \geqslant \SNR_0$,
\begin{equation}\label{EqEqualityHighSNR}
    \ds\sum_{k = 1}^2   u_{k}\left(\bs{p}^{(*,n)}\right) - \ds\sum_{k = 1}^2   u_{k}(\bs{p}^{\dagger}) \geqslant \delta,
\end{equation}
and $\delta \geqslant 0$.
}
\end{proposition}
The proof of Prop. \ref{PropPerformanceHighSNR} is given in App. \ref{AppendixB}. From the proof of Prop. \ref{PropPerformanceHighSNR}, it can be stated that in none of the games, players transmit simultaneously on the same channels. Now, from Prop. \ref{PropPerformanceLowSNR} and Prop. \ref{PropPerformanceHighSNR}, it can be concluded that at low SNR both games $\gameone$ and $\gametwo$ induce the same network spectral efficiency. On the contrary, the game $\gametwo$ always induce a higher or equal network spectral efficiency than the game $\gameone$ in the high SNR regime. This counter-intuitive result implies a Braess type paradox, as suggested in the beginning of this subsection.  

\subsection{The Case of Large Systems}\label{SecNonAtomicGame}

In this section, we deal with the games $\gametwo$ for the case of large networks, i.e., networks with a large number of transmitters.  
Within this scenario, the dominant parameter to analyze these games is the fraction of transmitters using a particular channel. As we shall see, contrary to the case of small number of transmitters and channels analyzed in the previous section, in the
case of large networks, each player becomes indifferent to the
action adopted by each of the other players. Here, each player is
rather concerned with the fractions of players simultaneously
playing the same action, i.e., using the same channel. Hence, one of the interesting issues to be solved is the determination of the repartition of the users between the different
channels at the NE. 

\noindent
As a first step towards
identifying the fractions of transmitters per channel at the NE, we
first re-write the potential (\ref{EqPotential}) as a function
of the vector $\bs{x}(\bs{p}) = \left(x_1(\bs{p}), \ldots, x_S(\bs{p})\right)$, where $x_s(\bs{p})$,
with $s \in \mathcal{S}$, denotes the fraction of players
transmitting over channel $s$ given the action profile $\bs{p}\in \mathcal{P}^{(b)}$. Hence,
\begin{equation}\label{EqFractions}
\begin{array}{cccc}
    \forall s \in \mathcal{S},&                            x_s(\bs{p}) &=& \frac{|\mathcal{K}_s(\bs{p})|}{K} \\
                              &\displaystyle\sum_{i = 1}^S x_i(\bs{p}) &=&  1,
\end{array}
\end{equation}
where $\mathcal{K}_s(\bs{p}) \subseteq \mathcal{K}$ is the set of players
 using channel $s$ given the action profile $\bs{p}\in\mathcal{P}^{(b)}$, i.e., $\mathcal{K}_s(\bs{p})= \lbrace k \in \mathcal{K} : p_{k,s} \neq 0 \rbrace$. Let  $b_s = \frac{B_s}{B}$ denote the
 fraction of bandwidth associated with channel $s$, such that $\sum_{s = 1}^{S} b_s = 1$.
Then, one can write the potential as follows
\begin{eqnarray}
\nonumber
\phi(\bs{p}) &=& \displaystyle\sum_{s=1}^{S} b_s
 \log_2\left( N_0 B_s  +    p_{\max} \displaystyle\sum_{k \in \mathcal{K}_s(\bs{p})} g_{k,s}\right)\\
\nonumber   &=&  S \log_2(K) + \displaystyle\sum_{s=1}^{S} b_s \log_2\left(\frac{No\,B_s}{K}   \right.\\
\label{EqPotentialFiniteCase} 
& + &\left.
x_s(\bs{p})\,p_{\max} \left(\frac{1}{\left| \mathcal{K}_s(\bs{p})\right|} \displaystyle\sum_{k \in \mathcal{K}_s(\bs{p})} g_{k,s}\right) \right).
\end{eqnarray}
Note that the term $S \log_2(K)$ in \eqref{EqPotentialFiniteCase} does not depend on the actions of the players. Thus, in the following, we drop it for the sake of simplicity.
We assume that the number of players $K$ and the available bandwidth $B$ grows to infinite at a constant rate $\mu > 0$, while the fractions $b_s$, for all $s \in \mathcal{S} $ are kept invariant. That is, the average bandwidth per transmitters is asymptotically constant,
\begin{equation}
\ds\lim_{K, \, B  \rightarrow \infty } \frac{B}{K} = \mu.
\end{equation}
Thus, under the assumption of large number of transmitters and for any action profile $\bs{p} \in \mathcal{P}^{(b)}$, it follows that,
\begin{equation}
\nonumber \forall s \in \mathcal{S}, \; \frac{1}{\left| \mathcal{K}_s(\bs{p})\right|}
 \displaystyle\sum_{k \in \mathcal{K}_s(\bs{p})}   g_{k,s}
 \stackrel{K \rightarrow \infty}{\longrightarrow} \displaystyle\int_{0}^{\infty} \lambda \d F_{g_s}(\lambda) = \Omega_s,
\end{equation}
where $F_{g_s}$ is the cumulative probability function associated with the channel gains over dimension $s$.
Hence, for all action profile $\bs{p} \in \mathcal{P}^{(b)}$ adopted by the players, maximizing the function $\phi\left(\bs{p}\right)$ in the asymptotic regime reduces to maximize the function $\tilde{\phi}\left(\bs{x}(\bs{p})\right) $,
\begin{equation}\nonumber
 \tilde{\phi}\left(\bs{x}(\bs{p})\right) = \displaystyle\sum_{s=1}^S   b_s  \log_2\left(\mu No\,b_s  +  x_s(\bs{p}) \,p_{\max}\,\Omega_s \right).
\end{equation}
That is, solving the OP,
\begin{equation}\nonumber
\left\lbrace
\begin{array}{cc}
    \displaystyle\max_{\bs{x} = \left( x_1, \ldots, x_S \right) \in \mathds{R}_+^{S}} & \displaystyle\sum_{s =1}^S   b_s \log_2\left(\mu N_0\, b_s   + x_s p_{\max} \Omega_s \right),\\
    \text{s.t.} & \displaystyle\sum_{i = 1}^{S}  x_i = 1 \; \text{ and  }  \;  \forall i \in \mathcal{S}, \; x_i \geqslant 0,
\end{array}
\right..
\end{equation}
The optimization problem above 
has a unique solution of the form,
\begin{equation}
    \forall s \in \mathcal{S}, \quad x_s = b_s \left[\frac{1}{\beta_{k}} - \frac{\mu N_0 }{ p_{\mathrm{max}} \Omega_s}\right]^+,
\end{equation}
where $\beta_k$ is Lagrangian multiplier to satisfy the optimization constraints. Interestingly,
in the case when all the channels are described with the same statistics, that is, $\forall s \in \mathcal{S}$, $F_{g_s}(\lambda) = F_{g}(\lambda)$, ($\forall s \in \mathcal{S}$, $\Omega_s = \Omega$) it holds that,
\begin{equation}\label{EqFractionSolution}
    \forall s \in \mathcal{S}, \quad x_s = \frac{B_s}{B}.
\end{equation}

\noindent
The above provides a very simple relation between the repartition of the users among the available channels in the asymptotic regime. Indeed, it can be implied that the number of transmitters using a given channel $s$ is proportional to the bandwidth allocated to the corresponding channel.
In particular, note that this result generalize the work in \cite{Belmega-isccsp-08}.

\noindent
To conclude on the usefulness of the large system analysis, let us consider the upper bound on the number of NE which is given by Proposition \ref{PropMultiplicityGb}. Let us normalize the upper bound on the number of pure NE $L$ in \eqref{EqL} by the total number of pure (channel selection) strategy profiles, and let us write, 
\begin{equation} 
\begin{array}{ccl}
\frac{ L}{S^K}
& < &  \frac{1}{S^K} \left(1 + \left(S - 1\right) 2^{K} \right).
\end{array}
\end{equation}
Now, for a sufficiently large number $K$, the following approximation holds,
\begin{equation} \label{nb-of-NE}
\begin{array}{ccl}
\frac{1}{S^K} \left(1 + \left(S - 1\right) 2^{K} \right) 
& \approx &  \left(S - 1\right) \left(\frac{2}{S}\right)^{K}.
\end{array}
\end{equation}
Although the number of pure NE in channel selection games my be conjectured to be large, it is in fact relatively small in the asymptotic regime. Indeed, (\ref{nb-of-NE}) indicates that when the number of users is large, the fraction of pure NE goes to zero whenever the number of channels is greater or equal to $3$. This result shows the difficulty of using methodologies such as the one proposed in Sec. \ref{SecDeterminationOfNEGameb} to study the set of NE or the difficulty of  
achieving equilibria using decentralized learning algorithms as proposed in \cite{Rose-Perlaza-CommMag-2011, Lasaulce-Tembine-Book-2011}.

\subsection{The Case of Successive Interference Cancellation}\label{SecSIC}

In order to analyze the case of successive interference cancellation (SIC), let us denote by $R_{\mathrm{SIC}}\left( \bs{p} \right)$ the NSE achieved with SIC assuming perfect decoding given the power allocation profile $\bs{p}$. That is,
\begin{eqnarray}\nonumber
R_{\mathrm{SIC}}\left( \bs{p} \right) & = & \displaystyle\sum_{s \in \mathcal{S}}\frac{B_s}{B} \log_2\left( 1 + \frac{\displaystyle\sum_{k = 1}^K p_{k,s} g_{k,s}}{\sigma^2_s}\right).
\end{eqnarray}
Now, the NSE with SIC $R_{\mathrm{SIC}}\left( \bs{p} \right)$ can be written in terms of the potential function \eqref{EqPotential} as follows,
\begin{eqnarray}\nonumber
R_{\mathrm{SIC}}\left( \bs{p} \right) & = & \displaystyle\sum_{s \in \mathcal{S}}\frac{B_s}{B}\log_2\left( \sigma^2_s  + \displaystyle\sum_{k \in \mathcal{K}} p_{k,s} g_{k,s}\right) \\
\nonumber
& & -  \displaystyle\sum_{s \in \mathcal{S}}\frac{B_s}{B}  \log_2\left(\sigma^2_s\right),\\
\label{EqSICsumRate}							
&= &	\phi(\bs{p}) - \displaystyle\sum_{s \in \mathcal{S}}\frac{B_s}{B}  \log_2\left(\sigma^2_s\right).
\end{eqnarray}
From \eqref{EqSICsumRate}, it can be immediately implied that maximizing the NSE of the parallel MAC under the assumption of perfect decoding is equivalent to maximize the potential function $\phi(\bs{p})$ in \eqref{EqPotential}. This observation and the notion of potential game (Def. \ref{DefPotentialGame}) lead us immediately to the following propositions (assume that player $k$ is decoded in the $k$-th place) :
\begin{proposition}\label{EqSICGameOne}
Let $\bs{p}^{\dagger} \in \mathcal{P}^{(a)}$ be the unique NE of the game $\gameone$. The maximum NSE of the network is achieved if for all $k \in \mathcal{K}$,  transmitter $k$ uses the power allocation vector $\bs{p}_k^{\dagger}$ and transmits with rate
\begin{equation}\nonumber
R_k^{\dagger} = \displaystyle\sum_{s \in \mathcal{S}}\frac{B_s}{B} \displaystyle\sum_{k = 1}^K \log_2\left( 1 + \frac{p_{k,s}^{\dagger} g_{k,s}}{\sigma^2_s + \displaystyle\sum_{j \in \mathcal{K}\setminus\lbrace 1, \ldots, k\rbrace} p_{j,s}^{\dagger} g_{j,s}}\right) .
\end{equation}
\end{proposition}
A similar result is obtained for the game $\gametwo$.

\begin{proposition}\label{EqSICGameTwo}
Let $\mathcal{P}^{+} \subset \mathcal{P}^{(b)}$ be the set of NE of the game $\gametwo$. The maximum NSE (achievable in the space $\mathcal{P}^{(b)}$) of the network is achieved if for all $k \in \mathcal{K}$,  transmitter $k$ uses the power allocation vector $\bs{p}_k^{+}$ and transmits with rate
\begin{equation}\nonumber
R_k^{+} = \displaystyle\sum_{s \in \mathcal{S}}\frac{B_s}{B} \displaystyle\sum_{k = 1}^K \log_2\left( 1 + \frac{p_{k,s}^{+} g_{k,s}}{\sigma^2_s + \displaystyle\sum_{j \in \mathcal{K}\setminus\lbrace 1, \ldots, k\rbrace} p_{j,s}^{+} g_{j,s}}\right) ,
\end{equation}
with $\bs{p}^{+} \in \mathcal{P}^{+}$ and
\begin{equation}
\forall \bs{p} \in \mathcal{P}^{+}, \quad \phi(\bs{p}^+) \geqslant \phi(\bs{p}).
\end{equation}
\end{proposition}
Prop. \ref{EqSICGameOne} and Prop. \ref{EqSICGameTwo} show that for achieving the maximum NSE the players must use the power allocation corresponding to the NE but the transmission rate must be adjusted according to the decoding order. This implies that each player needs to know both the SINR with and without SIC in order to set up its power allocation vector and its transmission rate, respectively. In the context of decentralized networks this amount of signaling is not always affordable. Thus, in the following numerical analysis, we only consider the more practical (and scalable) case of single-user decoding. More discussions on the unfeasibility of SIC in decentralized networks can be found in \cite{Mertikopoulos-JSAC-2011}.

\section{Numerical Examples}\label{SecNumericalExamples}

In the previous sections, a mathematical argument has been provided to show that at the low and high SNR regime, using a channel selection policy  yields a higher or equal NSE than using a water-filling power allocation policy. A formal proof has been provided for the case of $K = 2$ transmitters and $S=2$ channels.  Moreover, we highlight that a CS policy is evidently simpler than a PA policy in terms of implementation.
Unfortunately, providing a formal proof for an arbitrary number of transmitters $K$ and channels $S$ at a finite SNR becomes a hard task since it will require to calculate all the types of NE depending on the exact channel realizations. Hence, for the case of arbitrary parameters $K$, $S$, and SNR, we provide only numerical examples to give an insight of the general behavior. First, we evaluate the impact of the SNR for a network with a fixed number of transmitters and channels. Second, we evaluate the impact of the network load, i.e., the number of transmitters per channel for a given fixed SNR.

\subsection{Impact of the SNR $\frac{p_{\max}}{\sigma_2}$}

In Fig. \ref{FigLoadImpactAndSNRImpact} (left), we plot the network spectral efficiency as a function of the average SNR of the transmitters. Here, it is shown that in fully-loaded and over-loaded networks, i.e., $\eta = \frac{K}{S} \geqslant 1$, the gain in NSE obtained by using a discrete action set (game $\gameone$) increases with the SNR. Conversely, for weakly loaded networks $\eta < 1$, limiting the transmitters to use a single channel appears to be suboptimal as the SNR increases. This is basically because using only one channel, necessarily implies letting some interference-free channels unused. Interestingly, at low SNR, the NSE observed in both games is the same, independently of the load of the system. In both cases, high SNR and low SNR regime, the observed results are in line with Prop. \ref{PropPerformanceLowSNR} and Prop. \ref{PropPerformanceHighSNR}.

In Fig. \ref{FigNumberOfNE}, we plot the probability of observing a specific number of NE in the game $\gametwo$ for different values of SNR. In the first case (Fig. \ref{FigNumberOfNE} (left)) we consider $S = 2$ and $K = 3$, whereas in the second case (Fig. \ref{FigNumberOfNE} (right)), $K = 3$ and $S = 3$. Note that from Prop. \ref{PropMultiplicityGb}, the maximum number of NE is $4$ and $7$, respectively. However, only $3$ and $6$ NE are respectively observed in the simulations.
This mismatch is due to the fact that Prop. \ref{PropMultiplicityGb} relies only on the number of players and channels and takes into account only the distance (Def. \ref{DefDistance}) between two action profiles. Thus, it does not consider the utility function in \eqref{EqUtilityFunction} and the set of actions $\mathcal{P}^{(b)}$, for which there are some action profiles which are mutually exclusive of the set of NE. For instance, in the game $\gametwo$ with $K = 3$ and $S = 2$, the set of power allocations $p_{\max}\left(\bs{e}_1,\bs{e}_1,\bs{e}_1\right)$, $p_{\max}\left(\bs{e}_1,\bs{e}_2,\bs{e}_2\right)$, $p_{\max}\left(\bs{e}_2,\bs{e}_2,\bs{e}_1\right)$, and $p_{\max}\left(\bs{e}_2,\bs{e}_1,\bs{e}_2\right)$ are all at distance $2$ of each other. Nonetheless, if $p_{\max}\left(\bs{e}_1,\bs{e}_1,\bs{e}_1\right)$ is an equilibrium for a given set of channels, then the other three action sets are not NE for the same set of channels and vice-versa. Thus, only $3$ out of the $4$ candidates can be NE simultaneously. The exact number of NE can be determined following the method described in Sec. \ref{SecDeterminationOfNEGameb}, but it requires the complete knowledge of the channel gains. Prop. \ref{PropMultiplicityGb} aims at providing an estimation based only on the parameters $K$ and $S$.

\noindent
Finally, we remark that low SNR levels are associated with a unique NE (with high probability), whereas, high SNR levels are associated with multiple NE (with high probability). Note that this observation, at least for the case of $K = 2$ and $S = 2$, is inline with  Prop. \ref{PropPerformanceLowSNR} and Prop. \ref{PropPerformanceHighSNR}.

\subsection{Impact of the Number of Transmitters ($K$)}

In Fig. \ref{FigLoadImpactAndSNRImpact} (right), we plot the NSE as a function of the number of transmitters per channel, i.e., the system load $\eta = \frac{K}{S}$. Therein, one can observe that for weakly loaded systems $\eta < 1$, playing $\gameone$ always leads to higher NSE than playing $\gametwo$. This is natural since restricting the transmitters to use only one channel implies not using other channels which are interference-free, as $S > K$. On the contrary, for fully-loaded and over-loaded systems, the NSE of the game $\gameone$ is at least equal or better than that of the game $\gametwo$.
Indeed, the fact that for high system loads $\eta > 2$, the NSE obtained by playing the game $\gameone$ and $\gametwo$ become identical stems from the fact that under this condition the system becomes dominated by the interference.  
Finally, in Fig. \ref{FigFractions}, we show the fractions $x_s$ of transmitters using channel $s$, with $s \in
\mathcal{S}$, obtained by Monte-Carlo simulations and using (\ref{EqFractionSolution}) for a large network with an asymptotic ratio of players per channel equivalent to $\eta = 10$. Therein, it becomes clear that (\ref{EqFractionSolution}) is a precise estimation of the outcome of the games $\gameone$ and $\gametwo$ in the regime of large number of players.

\section{Conclusions}\label{SecConclusions}

In this paper, it is shown to what extent the equilibrium
analysis of the decentralized parallel MAC differs from those conducted for other channels like
Gaussian MIMO interference channels and fast fading MIMO MAC. In
particular, the special structure of parallel MAC  and the assumption of
single-user decoding at the receiver leads to the potential game property. The channel selection game was merely
introduced in the literature but not investigated in details as it
is in this paper. In particular, a graph-theoretic interpretation is
used to characterize the number of NE. In the case where the number of transmitters is large, the fraction between pure NE and the total number of action profiles is relatively small, which makes both the analysis and the achievability of the NE a challenging task. Now, from a design point of view, we provide theoretical results and numerical examples to show that a fully loaded network, when transmitters use only one channel, can be more efficient than its counterpart when all the channels can be exploited by the transmitters. Although all of
these results are encouraging about the relevance of game-theoretic
analyses of power allocation problems, important practical issues
have been deliberately ignored. For example, the impact of channel
estimation is not assessed at all. Also, it is important to conduct
a detailed analysis on the signaling cost involved by all the power
allocation algorithms arising from this game formulations to learn NE.

\appendices

\section{Proof of Proposition \ref{PropMultiplicityGb} }\label{AppProofOfPropMultiplicityGb}
In this appendix, we provide a proof of Proposition \ref{PropMultiplicityGb}, which establishes an upper bound for the number of NE of the game $\gameone$. Here, we exploit some basic tools from graph theory. Let
us index the elements of the action set $\mathcal{P}^{(b)}$ in any
given order using the index $n \in \mathcal{I} = \left\lbrace 1,
\ldots, S^K\right\rbrace$. Denote by $\bs{p}^{(n)}$ the $n$-th
element of the action set $\mathcal{P}^{(b)}$. We write each vector
$\bs{p}^{(n)}$ with $n \in \mathcal{I}$,  as  $\bs{p}^{(n)} =
\left(\bs{p}^{(n)}_1, \ldots, \bs{p}^{(n)}_{K}\right)$, where for
all $j \in \mathcal{K}$, $\bs{p}^{(n)}_{j} \in \mathcal{P}_j^{(b)}$.
Consider that each action profile $\bs{p}^{(n)}$ is associated with a vertex $v_n$ in a given
non-directed graph $G$. Each vertex $v_n$ is adjacent to the
$K(S-1)$ vertices associated with the action profiles resulting when
only one player deviates from the action profile $\bs{p}^{(n)}$,
i.e., if two vertices $v_{n}$ and $v_{m}$, with $(n,m) \in
\mathcal{I}^2$ and $n \neq m$, are adjacent, then there exists one and only one $k
\in \mathcal{K}$, such that
\begin{eqnarray}
\nonumber   \forall j \in \mathcal{K}\setminus\lbrace k \rbrace, &  \bs{p}^{(n)}_j = \bs{p}^{(m)}_j, \mbox{ and }
                                                                                                             &  \bs{p}^{(n)}_k \neq \bs{p}^{(m)}_k.
\end{eqnarray}
More precisely, the graph $G$ can be defined by the pair $G = \left( \mathcal{V}, \bs{A}\right)$, where the set $\mathcal{V} = \left\lbrace v_1, \ldots, v_{S^K}\right\rbrace$ (nodes) represents the $S^K$ possible actions profiles of the game and $\bs{A}$ (edges) is
a symmetric matrix (adjacency matrix of $G$) with dimensions $S^K \times S^K$ and entries defined as follows $\forall (n,m) \in \mathcal{I}^2$ and  $n\neq m$,
\begin{equation}
    a_{n,m} = a_{m,n} = \left\lbrace \begin{array}{cl} 1 & \text{ if } n \in \mathcal{V}_m\\
                                                        0 & \text{ otherwise },
                                     \end{array} \right.
\end{equation}
and $a_{n,n} = 0$ for all $n \in \mathcal{I}$, where the set $\mathcal{V}_n$ is the set of indices of the adjacent vertices of vertex $v_n$. In the following, we use the concept of distance between two vertices of the graph $G$. We define this concept using our notation:

\begin{definition}[Shortest Path]\label{DefDistance} \emph{The distance (shortest path) between vertices $v_n$ and $v_m$, with $(n,m) \in \mathcal{I}^2$ in a given non-directed graph $G =\left(\mathcal{V},A \right)$, denoted by $d_{n,m}(G) \in \mathds{N}$ is:
\begin{equation}
  d_{n,m}(G) = d_{m,n}(G) = \displaystyle\sum_{k = 1}^{K} \mathds{1}_{\left\lbrace \bs{p}^{(n)}_k \neq \bs{p}^{(m)}_k \right\rbrace}.
\end{equation}
}
\end{definition}

\noindent
Here, for any pair of action
profiles $\bs{p}^{(n)}$ and $\bs{p}^{(m)}$, with $(n,m) \in
\mathcal{I}^2$ and $n  \neq m$, we have that $\phi(\bs{p}^{(n)})
\neq  \phi(\bs{p}^{(m)})$ with probability one. This is because
channel gains are random variables drawn from continuous probability
distributions and thus,  $\Pr\left(\phi(\bs{p}^{(n)}) =
\phi(\bs{p}^{(m)})\left|\right. n \neq m \right) = 0$. Hence,
following Def. \ref{DefNE}, one can state that if the action profile
$\bs{p}^{(n^*)}$, with $n^* \in \mathcal{I}$, is an NE of the game
$\gametwo$, then, it follows that
\begin{equation}\label{EqCondNE}
\forall m \in \mathcal{V}_{(n^*)}, \quad    \phi(\bs{p}^{(n^*)}) > \phi(\bs{p}^{(m)}),
\end{equation}
and vice versa with probability one. However, several action profiles might simultaneously satisfy the condition (\ref{EqCondNE}), which is what we proof in the following.

\begin{proof}
From Prop. \ref{PropExistenceGb} it is ensured that $L \geqslant 1$.
Then, assume that a given action profile $\bs{p}^{(n)}$ (vertex
$v_n$) with $n \in \mathcal{I}$ is an NE. Given condition
(\ref{EqCondNE}), it follows that none of the vertices in the set
$\mathcal{V}_n$ is an NE. Nonetheless, if there exists another
action profile $\bs{p}^{(m)}$, with $m \in
\mathcal{I}\setminus\lbrace n \cup \mathcal{V}_n \rbrace$, which
satisfies (\ref{EqCondNE}), then $\bs{p}^{(m)}$ can be also an NE. Thus,
for the action profile $\bs{p}^{(m)}$, with $n \neq m$, to be an NE candidate, it must be (at least) at distance two of $\bs{p}^{(n)}$ and any other NE candidate, i.e., $d_{n,m}(G)=d_{m,n}(G) \in \lbrace 2, 4, \ldots, \hat{K}\rbrace$. An action profile at distance $\ell \in \lbrace 2, 4, \ldots, \hat{K}\rbrace$ from $\bs{p}^{(n)}$, is a vector where $\ell$ players have simultaneously deviated from $\bs{p}^{(n)}$. Hence, for each $\ell$-tuple of players, there always exists $S-1$ action profiles at distance $\ell$ from $\bs{p}^{(n)}$ and at distance $2$ from each other.
Thus, considering  the initial NE action profile
$\bs{p}^{(n)}$, there might exists at most
\begin{equation}\label{EqUpperBoundNE}
L \leqslant 1 + \ds\sum_{i \in \lbrace 2, 4, \ldots, \hat{K}\rbrace} \left(\begin{array}{c} K \\ i\end{array}\right)\left(S - 1\right)
\end{equation}
NE candidates. This establishes an
upper bound for $L$ and completes the proof.
\end{proof}

\section{Proof of Prop. \ref{PropNEGa}} \label{AppProf1}
In this appendix, we provide a proof for the Prop. \ref{PropNEGa}.  
The proof is separated in two steps. First, we show that a power allocation vector $\bs{p} = \left(\bs{p}_{1},\bs{p}_{2}\right) \in \mathcal{P}^{(a)}$ of the form
\begin{equation}
\nonumber   \bs{p}_{1} = \left(p_{11},p_{\max} - p_{11}\right) \text{ and } \bs{p}_{2} = \left( p_{\max} - p_{22},p_{22}\right),
\end{equation}
is not an NE of the game $\gameone$, when $p_{11} \in \left] 0, p_{\max}\right[$ and $p_{22}  \in
\left] 0, p_{\max}\right[$. 
Second, we show that if $\bs{p}$ is an NE, then, $\bs{p}$ is unique and satisfies that, $\bs{p} \in  \mathcal{P}^{\dagger}$, where
\begin{eqnarray}\nonumber
    \mathcal{P}^{\dagger} &=& \mathcal{P}^{(a)} \setminus \lbrace \bs{p} = \left(p_{11},p_{\max}-p_{11},p_{\max}-p_{22},p_{22}\right)  \\ 
    \nonumber
    &  & \in \mathds{R}_{+}^{4}:  p_{11}\in \left]0,p_{\max}\right[ \text{ and } p_{22} \in \left]0,p_{\max}\right[ \, \rbrace.
\end{eqnarray}
In the following, we use the notation $-c$ to denote the element other than $c$ in the binary set $\mathcal{C}$.

\begin{proof}
\textbf{First Step:} Assume that the action profile $\bs{p} =
\left(\bs{p}_1 ,\bs{p}_2 \right)$, with $\bs{p}_{1} = \left(p_{11}
,p_{12} \right)$ and $\bs{p}_{2} = \left(p_{21},p_{22}\right)$ is an
NE of the game $\gameone$, and assume that for all $(k,s) \in
\mathcal{K}\times \mathcal{S}$, $p_{k,s} > 0$, with strict
inequality. Then, from the best response correspondence, it holds that $\forall (k,s) \in \mathcal{K}\times\mathcal{S}$,
\begin{eqnarray}\label{EqPANE2x2}
p_{k,s}^{\dagger} & = & \left[\frac{1}{\beta_k} - \frac{\sigma^2 + g_{-k,s}p_{-k,s}^{\dagger}}{g_{k,s}}\right]^{+},
\end{eqnarray}
with $\beta_k$ a Langrangian multiplier chosen to satisfy \eqref{EqPowerConstraints}. Then, from \ref{EqPANE2x2}, it can be implied that
$\forall k \in \mathcal{K}$,
\begin{eqnarray}
    p_{k,s} & = & \frac{1}{\beta_k}  - \frac{\sigma^2 + p_{-k,s} g_{-k,s}}{g_{k,s}} >  0 \text{ and } \\
        p_{k,-s} & = & \frac{1}{\beta_k}  -  \frac{\sigma^2 + p_{-k,-s} g_{-k,-s}}{g_{k,-s}} > 0.
\end{eqnarray}
Then, from the fact that $\forall k \in \mathcal{K}$, $p_{k,s} + p_{k,-s} = p_{\max}$, we have that,
\begin{equation}\label{EqProof0}
\begin{array}{lcl}
p_{k,k} &=& \frac{1}{2}\left(p_{\max} - \frac{\sigma^2 + g_{-k,k}\left(p_{\max} - p_{-k,-k}\right)}{g_{k,k}} \right. \\ &+ & \left. \frac{\sigma^2 + g_{-k,-k} p_{-k,-k}}{g_{k,-k}}\right)\\
p_{k,-k} &=& p_{\max} - p_{k,k}.
\end{array}
\end{equation}
Using a matrix notation, the system of equations (\ref{EqProof0}) can be written as follows:
\begin{equation}\label{EqSysEq2x2}
     \bs{C} \left(\begin{array}{c}   p_{11}\\   p_{22} \end{array}\right) =   \bs{A},
\end{equation}
where, the matrix $\bs{C}$ is
\begin{equation}\nonumber
\bs{C} =  \left(\begin{array}{cc}  2g_{11}g_{12} &   -\left(g_{22}g_{11} + g_{21}g_{12}\right) \\
                                                      -\left(g_{22}g_{11} + g_{21}g_{12}\right) &   2g_{11}g_{12} \\
                \end{array} \right) 
\end{equation}
and, the vector $\bs{A}$ is
\begin{equation}\label{EqProof1}
    \bs{A} = \left( \begin{array}{c} p_{\max}g_{12}\left( g_{11} - g_{21} \right) + \sigma^2 \left(g_{11} - g_{12} \right)\\
                                                    p_{\max}g_{21}\left( g_{22} - g_{12} \right) + \sigma^2 \left(g_{22} - g_{21} \right) \end{array} \right).
\end{equation}
Note that the system of equations (\ref{EqSysEq2x2}) has a unique solution as long as the set of channels $\lbrace g_{11}, g_{12}, g_{21}, g_{22} \rbrace$ satisfies the condition $g_{12}g_{21} - g_{11}g_{22} \neq 0$.
Let us continue the analysis under the assumption that, $g_{12}g_{21} - g_{11}g_{22} \neq 0$ (the case where $g_{12}g_{21} - g_{11}g_{22} = 0$ is treated later). Then, the unique solution of (\ref{EqSysEq2x2}) is $\forall k \in \mathcal{K}$,
\begin{eqnarray}
\nonumber
    p_{k,k} & = &   \frac{p_{\max}g_{-k,k} \left( g_{k,-k} + g_{-k,-k} \right)}{g_{12}g_{21} - g_{11}g_{22}}\\
\nonumber    
    & + &   \frac{ \sigma^2 \left(g_{-k,k} + g_{-k,-k}\right)}{g_{12}g_{21} - g_{11}g_{22}},\\
\nonumber 
    p_{k,-k} & = & p_{\max} - p_{k,k}.
\end{eqnarray}
Note that if $g_{12}g_{21} - g_{11}g_{22} < 0$, then $\forall k \in \mathcal{K}$, $p_{k,k} < 0$, and, if $g_{12}g_{21} - g_{11}g_{22} > 0$, then $\forall k \in \mathcal{K}$, $p_{k,k} > p_{\max}$, which contradicts the initial power constraints (\ref{EqPowerConstraints}). Hence, any vector $\bs{p} = \left(\bs{p}_1 ,\bs{p}_2 \right)$, with $\bs{p}_{1} = \left(p_{11} ,p_{\max}-p_{11} \right)$ and $\bs{p}_{2} = \left(p_{\max}-p_{22},p_{22}\right)$, such that $\forall (k,s) \in \mathcal{K}\times\mathcal{S}$,  $0 < p_{k,s} < p_{\max}$ is not an NE for the game $\gameone$ when $g_{12}g_{21} - g_{11}g_{22} \neq 0$.
Assume now that $g_{12}g_{21} - g_{11}g_{22} =  0$, and let $\alpha = \frac{g_{21}}{g_{11}} = \frac{g_{22}}{g_{12}}$. Then,
the PA vector in \eqref{EqProof0} can be written as follows, for $k = 1$
\begin{equation}\label{EqFirstProof3}
\left\lbrace \begin{array}{lcl}
  p_{11} & = &   \alpha p_{22} + \frac{1}{2}\left(p_{\max} \left(1 - \alpha \right) + \sigma^2 \left(\frac{1}{g_{12}} - \frac{1}{g_{11}}\right) \right)\\
  p_{12} & = &   p_{\max} - p_{11}. \end{array}\right.
\end{equation}
and, for $k = 2$,
\begin{equation}\label{EqFirstProof4}
 \left\lbrace \begin{array}{lcl}
  p_{22} & = &  \frac{1}{2}\left(p_{\max} \left(1 + \frac{1}{\alpha} \right) + \sigma^2 \left(\frac{1}{g_{21}} - \frac{1}{g_{22}}\right) \right) + \frac{1}{\alpha} p_{11}\\
  p_{21} & = &   p_{\max} - p_{22}, \end{array}\right.
\end{equation}
Note that the first equations in both (\ref{EqFirstProof3}) and (\ref{EqFirstProof4}) are identical. Thus, we focus only on the first equation in (\ref{EqFirstProof3}). This implies that any PA vector, $\bs{p} = \left(\bs{p}_{1},\bs{p}_{2}\right)$, with  $\bs{p}_{1} = \left(p_{11},p_{\max} - p_{11}\right) \in \mathcal{P}_{1}^{(a)}$ and $\bs{p}_{2} = \left(p_{\max} - p_{22},p_{22}\right) \in \mathcal{P}_2^{(a)}$ satisfying the condition
\begin{equation}\label{EqConditions2x2}
p_{11} = \frac{1}{2}\left(p_{\max} \left(1 - \alpha \right) + \sigma^2 \left(\frac{1}{g_{12}} - \frac{1}{g_{11}}\right) \right) + \alpha p_{22}
\end{equation}
is an NE of the game $\gameone$ when $g_{12}g_{21} - g_{11}g_{22} =  0$ as long as $\forall (k,s) \in \mathcal{K}\times\mathcal{S}$,  $0 < p_{k,s} < p_{\max}$. For satisfying the latter conditions, it suffices to ensure that: $ 0 < p_{11} < p_{\max}$ when $p_{22} = 0$ and $p_{22} = p_{\max}$. Solving these inequalities leads to the following conditions over the channels:
\begin{description}
\item[$(i)$] $p_{11} > 0$, when $p_{22} = 0$, if  
\begin{equation}
\frac{1 + \frac{p_{\max}}{\sigma^2}g_{21}}{1 + \frac{p_{\max}}{\sigma^2} g_{12}} <  \frac{g_{11}}{g_{12}}\end{equation}.
\item[$(ii)$] $p_{11} > 0$, when $p_{22} = p_{\max}$, if  
\begin{equation}
\frac{1}{1 + \frac{p_{\max}}{\sigma^2}\left( g_{12} + g_{22} \right)} <  \frac{g_{11}}{g_{12}}.
\end{equation}
\item[$(iii)$] $p_{11} < p_{\max}$, when $p_{22} = 0$, if  
\begin{equation}
 \frac{g_{11}}{g_{12}} < 1 + \frac{p_{\max}}{\sigma^2}\left( g_{11} + g_{21} \right)
\end{equation}
\item[$(iv)$] $p_{11} < p_{\max}$, when $p_{22} = p_{\max}$, if  
\begin{equation}
\frac{g_{11}}{g_{12}} < 1 + \frac{p_{\max}}{\sigma^2}\left( g_{11} + g_{21} \right)
\end{equation}
\end{description}
Finally, we obtain that if the vector of channels $\bs{g} = \left(g_{11},g_{12},g_{21},g_{22}\right)$ satisfies that
\begin{eqnarray}
\nonumber
	\min\left(\frac{1}{1 + \frac{p_{\max}}{\sigma^2}(g_{12}+g_{22})}, \frac{1 + \frac{p_{\max}}{\sigma^2}g_{21}}{1 + \frac{p_{\max}}{\sigma^2}g_{12}}\right) & < & \alpha  \text{ and }\\
\nonumber	
  \max\left(  1 + \frac{p_{\max}}{\sigma^2}\left( g_{11} + g_{21} \right),  \frac{1 + \frac{p_{\max}}{\sigma^2}g_{11}}{1 + \frac{p_{\max}}{\sigma^2}g_{22}} \right) & > &\alpha,
\end{eqnarray}
that is,
\begin{eqnarray}\nonumber
\frac{1}{1 + \frac{p_{\max}}{\sigma^2}(g_{12}+g_{22})} < & \alpha & <  1 + \frac{p_{\max}}{\sigma^2}\left( g_{11} + g_{21} \right),
\end{eqnarray}
then any vector $\bs{p} = \left(\bs{p}_{1},\bs{p}_{2}\right)$, with  $\bs{p}_{1} = \left(p_{11},p_{\max} - p_{11}\right) \in \mathcal{P}_{1}^{(a)}$ and $\bs{p}_{2} = \left(p_{\max} - p_{22},p_{22}\right) \in \mathcal{P}_2^{(a)}$ satisfying the condition \eqref{EqConditions2x2} is an NE of the game $\gameone$.
Note that infinitely many PA vectors might satisfy \eqref{EqConditions2x2}, which implies infinitely many NE. However, since the channels are realizations of random variables drawn from a continuous distribution, the probability of observing a realization such that $g_{12}g_{21} - g_{11}g_{22} =  0$ is zero.
Thus, with probability one, any vector $\bs{p} = \left(\bs{p}_1 ,\bs{p}_2 \right)$, with $\bs{p}_{1} = \left(p_{11} ,p_{\max}-p_{11} \right)$ and $\bs{p}_{2} = \left(p_{\max}-p_{22},p_{22}\right)$, such that $\forall (k,s) \in \mathcal{K}\times\mathcal{S}$,  $0 < p_{k,s} < p_{\max}$ is not an NE for the game $\gameone$.

\textbf{Second Step: } Consider that $\bs{p}^{\dagger} =
\left(\bs{p}_{1}^{\dagger},\bs{p}_{2}^{\dagger}\right) \in \mathcal{P}^{(a)}$ is an
NE. Then, it must follow that $\bs{p}^{\dagger} \in \mathcal{P}^{\dagger}$,
where,
\begin{equation}
\begin{array}{lcl}
    \mathcal{P}^{\dagger}&=& \mathcal{P}\setminus\lbrace \bs{p} = \left(p_{11},p_{\max}-p_{11},p_{\max}-p_{22},p_{22}\right) \\
    \nonumber
    & &\in \mathds{R}_{+}^{4}: p_{11}\in \left]0,p_{\max}\right[ \text{ and } p_{22} \in \left]0,p_{\max}\right[ \, \rbrace\\
                      &=& \displaystyle \bigcup_{n = 1}^{8} \mathcal{P}^{\dagger}_{i},
\end{array}
\end{equation}
where the sets $\mathcal{P}^{\dagger}_{n} \subset \mathcal{P}^{(a)}$, for all $n \in \lbrace 1, \ldots, 8\rbrace$ are described as follows. The singletons $\mathcal{P}^{\dagger}_{1} =   \lbrace  \bs{p} = \left(p_{\max},0,0,p_{\max}\right)\rbrace$, $\mathcal{P}^{\dagger}_{2} =  \lbrace  \bs{p} = \left(p_{\max},0,p_{\max},0\right)\rbrace$, $\mathcal{P}^{\dagger}_{3}=   \lbrace  \bs{p} = \left(0,p_{\max},0,p_{\max}\right)\rbrace$, $\mathcal{P}^{\dagger}_{4} = \lbrace  \bs{p} = \left(0,p_{\max},p_{\max},0\right)\rbrace$ and the convex non-closed sets,
\begin{eqnarray}
\nonumber   \mathcal{P}^{\dagger}_{5}& = &  \lbrace  \bs{p} = \left(p_{11},p_{\max}-p_{11},p_{\max}-p_{22},p_{22}\right) \\ 
\nonumber
& & \in \mathds{R}_{+}^{4}: p_{11} = p_{\max}, \text{ and } p_{22} \in \left]0,p_{\max}\right[ \rbrace,\\
\nonumber   \mathcal{P}^{\dagger}_{6}& = &  \lbrace  \bs{p} = \left(p_{11},p_{\max}-p_{11},p_{\max}-p_{22},p_{22}\right) \\
\nonumber
&  & \in \mathds{R}_{+}^{4}: p_{11} \in \left]0,p_{\max}\right[ \text{ and } p_{22} = p_{\max} \rbrace,\\
\nonumber   \mathcal{P}^{\dagger}_{7}& = &  \lbrace  \bs{p} = \left(p_{11},p_{\max}-p_{11},p_{\max}-p_{22},p_{22}\right) \\
\nonumber
& & \in \mathds{R}_{+}^{4}: p_{11} \in \left]0,p_{\max}\right[ \text{ and } p_{22} = 0 \rbrace,\\
\nonumber   \mathcal{P}^{\dagger}_{8}& = &  \lbrace  \bs{p} = \left(p_{11},p_{\max}-p_{11},p_{\max}-p_{22},p_{22}\right) \\
\nonumber
& & \in \mathds{R}_{+}^{4}: p_{11} = 0, \text{ and } p_{22} \in \left]0,p_{\max}\right[ \rbrace.
\end{eqnarray}

In the following, we identify the conditions over the channel vector $\bs{g} = \left(g_{11},g_{12},g_{21},g_{22}\right)$ such that each $\bs{p}^{\dagger} \in \mathcal{P}^{\dagger}_n$, with $n \in \lbrace 1, \ldots, 8 \rbrace$ is an NE.

\noindent
Assume that $\bs{p}^{\dagger} \in \mathcal{P}^{\dagger}_{8}$, i.e., $\bs{p}_{1}^{\dagger} = \left(0,p_{\max}\right)$ and $\bs{p}_{2}^{\dagger} = \left(p_{\max}-p_{22}^{\dagger},p_{22}^{\dagger}\right)$, with $p_{22}^{\dagger} \in \left]0,p_{\max}\right[$. Then, from (\ref{EqPANE2x2}) with $k = 2$, we have that:
\begin{eqnarray}
    p_{21}^{\dagger} & = & \frac{1}{\beta_2} - \frac{\sigma^2}{g_{21}} > 0\text{ and } \\
    p_{22}^{\dagger} &= & \frac{1}{\beta_2} - \frac{\sigma^2 + g_{12}p_{\max}}{g_{22}} > 0.
\end{eqnarray}
Then, since $p_{21}^{\dagger} + p_{22}^{\dagger} = p_{\max}$, we have that $\frac{1}{\beta_2} = \frac{1}{2}\left(p_{\max} + \frac{\sigma^2 + p_{\max} g_{12}}{g_{22}} + \frac{\sigma^2}{g_{21}}\right)$, and thus,
\begin{equation}\label{EqProof4}
     p_{22}^{\dagger} = \frac{1}{2}\left(p_{\max} - \frac{\sigma^2 + g_{12}p_{\max}}{g_{22}} + \frac{\sigma^2}{g_{21}}\right),
\end{equation}
where, it must satisfy that $0 < p_{22}^{\dagger} < p_{max}$. The inequality $p_{22}^{\dagger} > 0$, holds only if $\frac{g_{21}}{g_{22}} < \frac{1 + \frac{p_{\max}g_{21}}{\sigma^2}}{1 + \frac{p_{\max}g_{12}}{\sigma^2}}$,
whereas the inequality $p_{22}^{\dagger} < p_{\max}$ holds only if $\frac{g_{21}}{g_{22}} > \frac{1}{1 + \SNR\left(g_{12} + g_{22}\right)}$.
Similarly, from (\ref{EqPANE2x2}) with $k = 1$, we have that given $p_{22}^{\dagger}$, in order to obtain $p_{11}^{\dagger} = 0 $ and $p_{12}^{\dagger} = p_{\max}$, it must hold that
\begin{eqnarray}\nonumber
p_{11} & = &\frac{1}{\beta_1} -  \frac{\sigma^2 + g_{21}\left(p_{\max}-p_{22}^{\dagger}\right)}{g_{11}} \leqslant 0 \text{ and } \\
\label{EqProof5}
p_{12} &=& \frac{1}{\beta_1} -  \frac{\sigma^2 + g_{22} p_{22}^{\dagger}}{g_{12}} \geqslant p_{\max}.
\end{eqnarray}
Hence, by doing $p_{12}  - p_{11}$ in \eqref{EqProof5}, we obtain that:
\begin{equation}\label{EqProof7}
\frac{\sigma^2 + g_{21}\left(p_{\max}-p_{22}^{\dagger}\right)}{g_{11}} - \frac{\sigma^2 + g_{22} p_{22}^{\dagger}}{g_{12}} \leqslant p_{\max}.
\end{equation}
Then, by replacing (\ref{EqProof4}) in (\ref{EqProof7}), we obtain that the condition (\ref{EqProof5}) are satisfied only if the channels satisfy that:
\begin{equation}\label{EqProof7a}
    \frac{g_{11}}{g_{12}} \leqslant \frac{g_{21}}{g_{22}}.
\end{equation}
Hence, we can conclude that whenever the vector $\bs{g} = \left(g_{11}, g_{12}, g_{21}, g_{22}\right) \in \mathcal{B}_8$, the NE is of the form $\left(p_{11},p_{\max}-p_{11},p_{\max}-p_{22},p_{22}\right)$, with $p_{11} = 0$ and $p_{22} = \frac{1}{2}\left(p_{\max} - \frac{\sigma^2 + g_{12}p_{\max}}{g_{22}} + \frac{\sigma^2}{g_{21}}\right)$. Now, assuming that $\bs{p}^{\dagger} \in \mathcal{P}^{\dagger}_n$, with $n \in \lbrace 1,\ldots,7\rbrace$, leads to the conditions of the other types of NE, i.e., the corresponding sets $\mathcal{B}_n$, such that whenever $\bs{g} \in \mathcal{B}_n$ then $\bs{p}^{\dagger} \in \mathcal{P}^{\dagger}_n$. It is important to note that, for any particular vector $\bs{g} = \left(g_{11}, g_{12}, g_{21}, g_{22}\right) \in \mathcal{R}^4$, there exists, with probability one, only one set $\mathcal{B}_n$ which satisfies that $\bs{g} \in \mathcal{B}_n$. This is basically because for all $(n,m) \in \lbrace 1, \ldots, 4\rbrace^2$, with $n \neq m$, it follows that $\mathcal{B}_n \cap \mathcal{B}_m  = \emptyset$. Now, for all $(n,m) \in \lbrace 5, \ldots, 8\rbrace^2$, with $n \neq m$, it follows that $\mathcal{B}_n \cap \mathcal{B}_m  \subset \lbrace \bs{g} = \left(g_{11},g_{12},g_{21},g_{22}\right) \in \mathds{R}^{4}: g_{11}g_{22} = g_{12}g_{21}\rbrace$ and observing a channel realization $\bs{g}$, such that, $g_{11}g_{22} = g_{12}g_{21}$ is a zero probability event, since all channel gains are drawn from continuous probability distributions. Thus, with probability one, the game $\gameone$ has unique NE. This completes the proof.
\end{proof}

\section{Proof of Prop. \ref{PropPerformanceLowSNR}}\label{ProofPerformanceLowSNR}

In this appendix, we provide a proof of the Prop.  \ref{PropPerformanceLowSNR}. The Prop.  \ref{PropPerformanceLowSNR} basically states that at low SNR regime if an action profile $\bs{p}$ is a NE of the game $\gameone$, then it is also a NE of the game $\gametwo$ and it is unique.
The proof follows from the fact that in the asymptotic regime, i.e., $\SNR \rightarrow 0$, the set $\mathcal{A}_n$ and the set $\mathcal{B}_n$ become identical, when $n \in \lbrace 1, \ldots, 4\rbrace$. Moreover, the sets $\mathcal{A}_{m}$, with $m\in \lbrace 5, \ldots 8\rbrace$ become empty. The sets $\mathcal{A}_1, \ldots, \mathcal{A}_4$ and the sets $\mathcal{B}_1, \ldots, \mathcal{B}_8$ are given by Prop. \ref{PropNEGb} and Prop. \ref{PropNEGa}, respectively.
Then, 
\begin{eqnarray}
\nonumber \ds\lim\limits_{\SNR \rightarrow 0} \mathcal{A}_1 = \ds\lim\limits_{\SNR \rightarrow 0} \mathcal{B}_1 & = & \lbrace \bs{g} \in \mathds{R}_+^4: \, \frac{g_{11}}{g_{12}} \geqslant 1 \text{ and } \\ 
\nonumber
& & \frac{g_{21}}{g_{22}} \leqslant 1 \rbrace\\
\nonumber
\ds\lim\limits_{\SNR \rightarrow 0} \mathcal{A}_2 = \ds\lim\limits_{\SNR \rightarrow 0} \mathcal{B}_2 & = &  \lbrace \bs{g} \in \mathds{R}_+^4: \, \frac{g_{11}}{g_{12}} \geqslant 1 \text{ and } \\
\nonumber
& & \frac{g_{21}}{g_{22}} \geqslant 1 \rbrace\\
\nonumber
\ds\lim\limits_{\SNR \rightarrow 0} \mathcal{A}_3  = \ds\lim\limits_{\SNR \rightarrow 0} \mathcal{B}_3 &=& \lbrace \bs{g} \in \mathds{R}_+^4: \, \frac{g_{11}}{g_{12}} \leqslant 1 \text{ and } \\
\nonumber
& & \frac{g_{21}}{g_{22}} \leqslant 1 \rbrace\\
\nonumber
\ds\lim\limits_{\SNR \rightarrow 0} \mathcal{A}_4  = \ds\lim\limits_{\SNR \rightarrow 0} \mathcal{B}_4  & = & \lbrace \bs{g} \in \mathds{R}_+^4: \, \frac{g_{11}}{g_{12}} \leqslant 1 \text{ and } \\
\nonumber 
& & \frac{g_{21}}{g_{22}} \geqslant 1 \rbrace   \end{eqnarray}
and moreover,
\begin{eqnarray}
\ds\lim\limits_{\SNR \rightarrow 0}\mathcal{A}_5 & = & \lbrace \bs{g} \in \mathds{R}_+^4: \; \frac{g_{11}}{g_{12}} = 1 \text{ and } \frac{g_{21}}{g_{22}} = 1  \rbrace,   \nonumber
\end{eqnarray}
and
\begin{equation}\label{EqLowSNR5}
\forall n \in \lbrace 6, \ldots, 8\rbrace, \quad    \ds\lim\limits_{\SNR \rightarrow 0} \mathcal{A}_n    = \emptyset.
\end{equation}
Now, since the sets $\mathcal{A}_1, \ldots, \mathcal{A}_4$ or the sets $\mathcal{B}_1, \ldots, \mathcal{B}_4$ cover, in the asymptotic regime, all the space of vectors $\bs{g}$ and both $\mathcal{A}_n$ and $\mathcal{B}_n$ determine a unique NE in  the game $\gametwo$ and $\gameone$, respectively, it follows that the NE of both games  is identical in the asymptotic regime. The uniqueness of the NE in the game
   $\gameone$ holds with probability one, independently of the
 SNR level (Prop. \ref{PropNEGa}). In the game $\gametwo$, the NE is not unique if and only if $\bs{g} \in \mathcal{A}_5$. Nonetheless, since for all $(k,s) \in \mathcal{K}\times \mathcal{S}$, $g_{k,s}$ is a realization of a random variable drawn from a continuous probability distribution, we have that
\begin{equation}
   \Pr\left(\bs{g} \in \mathcal{A}_5 \right)  = 0.
\end{equation}
Thus, with probability one, the NE of the game $\gametwo$ is unique at high SNR regime, which completes the proof.

\section{Proof of Prop. \ref{PropPerformanceHighSNR}} \label{AppendixB}

In this appendix, we provide the proof of Prop.
\ref{PropPerformanceHighSNR}, which states that at the high SNR
regime there always exists an NE action profile in the game
$\gametwo$, which leads to an equal or better global performance than the
unique NE of the game $\gameone$.
Before we start, we introduce two lemmas which are used in the proof.

\begin{lemma}\label{LemmaNEG1HighSNR} \emph{In the high SNR regime, the game $\gameone$ possesses a unique NE, which can be of six different types depending on the channel realizations $\channelset$:
\begin{itemize}
    \item Equilibrium $1$: if $\bs{g} \in \mathcal{B'}_1 = \lbrace \bs{g} \in \mathds{R}_+^4: g_{22} \geqslant  g_{12}, \; \text { and } g_{21} \leqslant g_{11} \rbrace$, then, $p_{11}^{\dagger} =  p_{\max}$ and $p_{22}^{\dagger} =  p_{\max}$.
    \item Equilibrium $4$: if $\bs{g} \in \mathcal{B'}_4 \lbrace \bs{g} \in \mathds{R}_+^4: g_{11} \leqslant g_{21}, \; \text { and } g_{12} \geqslant g_{22} \rbrace$, then, $p_{11}^{\dagger} = 0$ and $p_{22}^{\dagger} =  0$.
    \item Equilibrium $5$: if $\bs{g} \in \mathcal{B'}_{5} = \lbrace \bs{g} \in \mathds{R}_+^4: \frac{g_{11}}{g_{12}}  \geqslant  \frac{g_{21}}{g_{22}}, \; \text { and } g_{21} > g_{11}\rbrace$, then, $p_{11}^{\dagger} = p_{\max}$ and  $p_{22}^{\dagger} = \frac{1}{2}\left(p_{\max} - \frac{\sigma^2}{g_{22}} + \frac{\sigma^2+g_{11} p_{\max}}{g_{21}}\right)$.
    \item Equilibrium $6$: if $\bs{g} \in \mathcal{B'}_6 = \lbrace \bs{g} \in \mathds{R}_+^4: \frac{g_{11}}{g_{12}}  \geqslant  \frac{g_{21}}{g_{22}}, \; \text { and } g_{22}  <  g_{12} \rbrace$, then, $p_{11}^{\dagger} = \frac{1}{2}\left(p_{\max} - \frac{\sigma^2}{g_{11}} + \frac{\sigma^2 + p_{\max}g_{22}}{g_{12}}\right)$ and $ p_{22}^{\dagger} =  p_{\max}$.   %
    \item Equilibrium $7$: if $\bs{g} \in \mathcal{B'}_{7} = \lbrace \bs{g} \in \mathds{R}_+^4: \frac{g_{11}}{g_{12}} \leqslant \frac{g_{21}}{g_{22}}, \; \text { and }  g_{11} > g_{21} \rbrace$, then, $p_{11}^{\dagger} = \frac{1}{2}\left(p_{\max} - \frac{\sigma^2 + p_{\max}g_{21}}{g_{11}} + \frac{\sigma^2}{g_{12}}\right)$ and $p_{22}^{\dagger}  = 0$.
\item Equilibrium $8$: if $\bs{g} \in \mathcal{B'}_{8}  = \lbrace \bs{g} \in \mathds{R}_+^4: \frac{g_{11}}{g_{12}}  \leqslant \frac{g_{21}}{g_{22}}, \; \text { and }  g_{12} <  g_{21} \rbrace$, then, $p_{11}^{\dagger} = 0$ and $p_{22}^{\dagger} = \frac{1}{2}\left(p_{\max} - \frac{\sigma^2 + g_{12} p_{\max}}{g_{22}} + \frac{\sigma^2}{g_{21}}\right)$.
    \end{itemize}
    }
\end{lemma}

The proof of lemma \ref{LemmaNEG1HighSNR} follows the same reasoning of the proof of Prop \ref{PropNEGa}. Here, $\forall n \in \lbrace 1, \ldots, 8 \rbrace$, $\mathcal{B}'_n = \ds\lim_{\SNR \rightarrow \infty} \mathcal{B}_n$, where the sets $\mathcal{B}_1, \ldots, \mathcal{B}_8$ are given by Prop. \ref{PropNEGa}.

\noindent
In the following lemma, we describe the set of NE of the game $\gametwo$ in the high SNR
regime.

\begin{lemma}\label{LemmaNEG2HighSNR}\emph{In the high SNR regime, the game $\gametwo$ always possesses two NE action profiles:
\begin{equation}
\bs{p}^{*,1}_1 = (0, p_{\max}) \text{ and } \bs{p}^{*,1}_2 = (p_{\max},0)
\end{equation}
and
\begin{equation}
\bs{p}^{*,4}_1 = (p_{\max},0) \text{ and } \bs{p}^{*,4}_2 = (0,p_{\max}),
\end{equation}
independently of the channel realizations.
}
\end{lemma}
\begin{proof} In the high SNR, i.e., $\SNR \rightarrow \infty$, the sets $\mathcal{A}_1, \ldots, \mathcal{A}_4$, given by Prop. \ref{PropNEGb}, become the following sets,
\begin{eqnarray}
\ds\lim\limits_{\SNR \rightarrow +\infty} \mathcal{A}_1 = \ds\lim\limits_{\SNR \rightarrow +\infty} \mathcal{A}_4 & = &  \mathds{R}_+^4\\
\ds\lim\limits_{\SNR \rightarrow +\infty} \mathcal{A}_2 = \ds\lim\limits_{\SNR \rightarrow +\infty}  \mathcal{A}_3 & = &
\emptyset.
\label{EqHighSNR4}
\end{eqnarray}

Thus, one can immediately imply that 
$$\textstyle\Pr\left(\bs{g} \in \textstyle\lim\limits_{\SNR \rightarrow +\infty} \mathcal{A}_2\right) = \textstyle\Pr\left(\bs{g} \in \textstyle\lim\limits_{\SNR \rightarrow +\infty} \mathcal{A}_3\right) = 0,$$ 
and,
 $$\textstyle\Pr\left(\bs{g} \in \textstyle\lim\limits_{\SNR \rightarrow +\infty} \mathcal{A}_1\right) = \textstyle\Pr\left(\bs{g} \in \textstyle\lim\limits_{\SNR \rightarrow +\infty} \mathcal{A}_4\right) = 1.$$ 
Hence, from Prop. \ref{PropNEGb}, we imply that both $\bs{p}^{(*,1)}$ and $\bs{p}^{(*,4)}$ are NE action profiles of the game $\gametwo$ in the high SNR regime regardless of the exact channel realizations $\channelset$, which completes the proof.
\end{proof}
From Lemma \ref{LemmaNEG1HighSNR} and Lemma \ref{LemmaNEG2HighSNR}, it is easy to see that if $\bs{g} = \channelvector \in \mathcal{B'}_{n}$, with $n \in \lbrace 1, 4\rbrace$, then (\ref{EqEqualityHighSNR}) holds since $\bs{p}^{\dagger}$ and at least one of the NE action profiles $\bs{p}^{*,n}$, with $n \in \lbrace 1, 4\rbrace$ are identical. In the cases where $\bs{g} = \channelvector \in \mathcal{B'}_{n}$, with $n \in \lbrace 5, \ldots, 8\rbrace$, we prove by inspection that in all the cases condition (\ref{EqEqualityHighSNR}) always holds for both NE action profiles $\bs{p}^{(*,1)}$ and $\bs{p}^{(*,4)}$.
For instance, assume that $\bs{g} \in \mathcal{B'}_{5}$. Then, we have that the unique NE of the game $\gameone$ is $\bs{p}^{\dagger} = (p_{11}^{\dagger}, p_{\max} - p_{11}^{\dagger}, p_{\max} - p_{22}^{\dagger},p_{22}^{\dagger})$, with $p_{11}^{\dagger} = p_{\max}$ and $p_{22}^{\dagger} = \frac{1}{2}\left(p_{\max} + \frac{\sigma^2 + p_{\max}g_{11}}{g_{21}} - \frac{\sigma^2}{g_{22}}\right)$ (See Lemma \ref{LemmaNEG1HighSNR}). Define the function $\psi: \mathds{R}_+ \rightarrow \mathds{R}_+$  as follows: $\psi(x) = 1 + \SNR x$, with $\SNR = \frac{p_{\max}}{\sigma^2}$, and denote by $\Delta_1\left(\SNR\right)$, the difference between the NSE  achieved by playing $\gameone$ and $\gametwo$, with respect to the NE $\bs{p}^{*,1}$ at SNR level $\SNR$, i.e.,
\begin{eqnarray}
\nonumber   \Delta_1\left(\SNR\right) & = & u_1\left(\bs{p}^{*,1}\right) + u_2\left(\bs{p}^{*,1}\right) - \\
\nonumber
& & \left( u_1\left(\bs{p}^{\dagger}\right) + u_2\left(\bs{p}^{\dagger}\right) \right)\\
\nonumber                   
& = & 2 \log_2 \left(2\right) - 2 \log_2\left(1 + \frac{g_{21}}{g_{22}} \frac{\psi(g_{22})}{\psi(g_{11})}\right) -  \\
\nonumber
& & \log_2\left(1 + \frac{g_{22}}{g_{21}} \frac{\psi(g_{11})}{\psi(g_{22})}\right) +  \\
\nonumber
& & \log_2\left(\frac{g_{21}}{g_{22}} +  \psi(g_{21} - g_{11}) \right).
\end{eqnarray}
Note that if $\bs{g} \in \mathcal{B'}_{5}$, then $g_{21} > g_{11}$. Hence,
\begin{eqnarray}
\nonumber       \ds\lim_{\SNR  \rightarrow \infty} \Delta_1\left(\SNR\right)
\nonumber                     
                     & = & \infty,
\end{eqnarray}
which justifies (\ref{EqEqualityHighSNR}).
Similarly, denote $\Delta_4\left(\SNR\right)$, the difference between the NSE achieved by playing $\gameone$ and $\gametwo$, with respect to the NE $\bs{p}^{*,4}$, i.e.,
\begin{eqnarray}
\nonumber   \Delta_4\left(\SNR\right) & = & \textstyle u_1\left(\bs{p}^{*,4}\right) + u_2\left(\bs{p}^{*,4}\right) - \\
& & \left( u_1\left(\bs{p}^{\dagger}\right) + u_2\left(\bs{p}^{\dagger}\right) \right)\nonumber\\
\nonumber    
& = & 2\log_2 \left(2\right) - 2 \log_2\left(1 + \frac{g_{22}}{g_{21}} \frac{\psi(g_{21})}{\psi(g_{12})}\right) \\
& & -  \log_2\left(1 + \frac{g_{21}}{g_{22}} \frac{\psi(g_{12})}{\psi(g_{21})}\right) \\
& & + \log_2\left(\frac{g_{22}}{g_{21}} +  \psi(g_{22} - g_{12}) \right).
\end{eqnarray}
Note that if $\bs{g} \in \mathcal{B'}_{5}$, then $g_{22} > g_{12}$. Hence,
\begin{eqnarray}
\nonumber       \ds\lim_{\SNR  \rightarrow \infty} \Delta_4\left(\SNR\right)
                     & = &\infty,
\end{eqnarray}
which justifies (\ref{EqEqualityHighSNR}). Hence, one can imply that in the high SNR regime both NE action profiles $\bs{p}^{*,1}$ and $\bs{p}^{*,2}$, satisfy (\ref{EqEqualityHighSNR}) when $\bs{g} \in \mathcal{B'}_5$. The same result as the one obtained when $\bs{g} \in \mathcal{B'}_5$, is also obtained when  $\bs{g} \in \mathcal{B'}_n$, with $n \in \lbrace 6, \ldots, 8\rbrace$, which completes the proof.

\bibliographystyle{IEEEtran}
\bibliography{GT}

\begin{figure}
\centering
\tiny
\graphicspath{{Figures/}}
\def\svgwidth{1\linewidth}
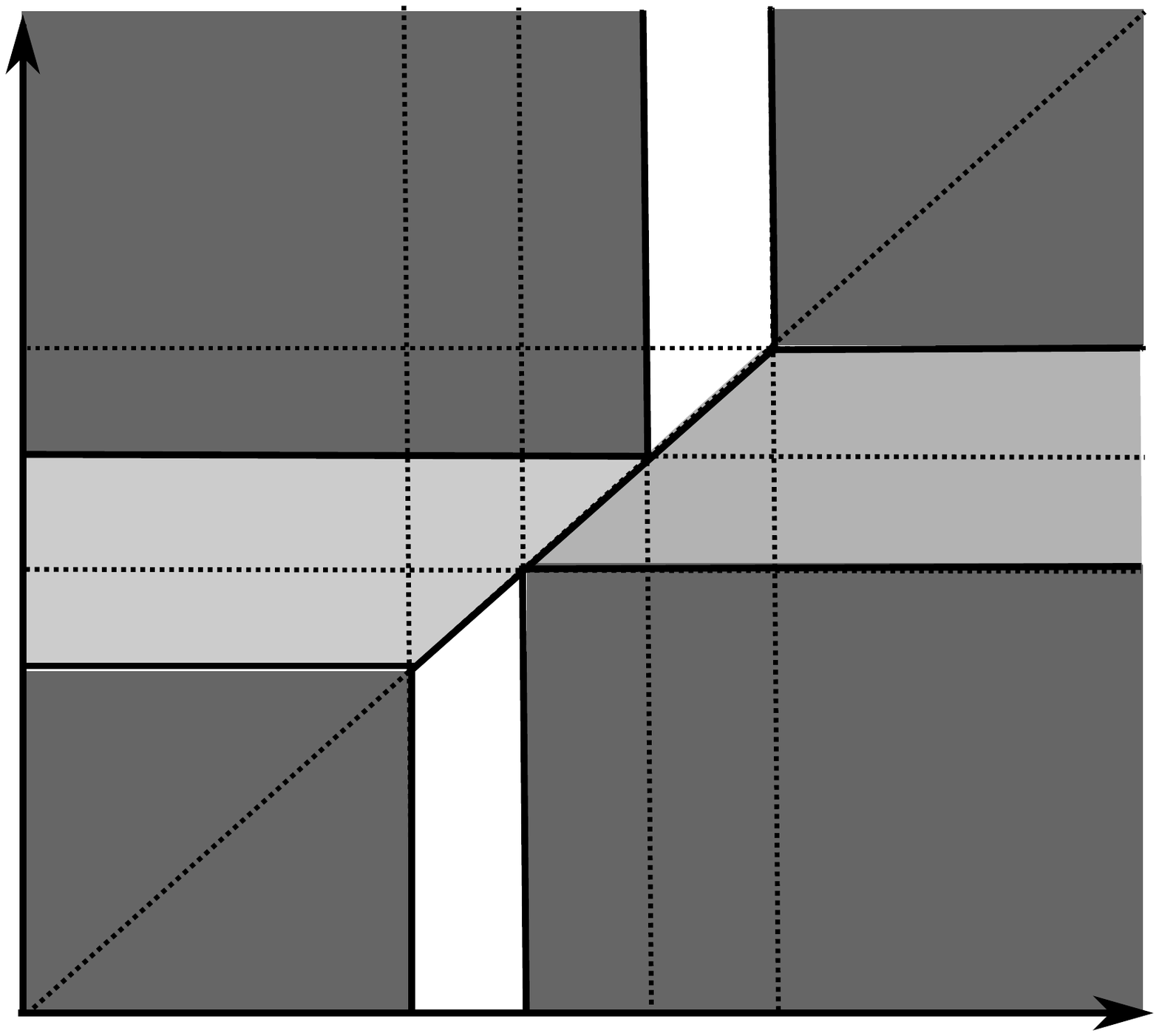
\caption{Nash equilibrium action profiles as a function of the channel ratios $\frac{g_{11}}{g_{12}}$ and $\frac{g_{21}}{g_{22}}$ for the two-player-two-channel game $\gameone$ (left) and $\gametwo$ (right), respectively. The function $\psi: \mathds{R}_+ \rightarrow \mathds{R}_+$ is defined as follows: $\psi(x) = 1 + \SNR  \, x$. The best response function $\BR_k(\bs{p}_{-k})$, for all $k \in \mathcal{K}$, is defined by (\ref{EqNEGameA}). Here, it has been arbitrarly assumed that $\frac{\psi(g_{21})}{\psi(g_{12})} < \frac{\psi(g_{11})}{\psi(g_{22})}$.}
\label{FigNashRegionsG1}
\end{figure}
\begin{figure}
\centering
\tiny
\def\svgwidth{1\linewidth}
\graphicspath{{Figures/}}
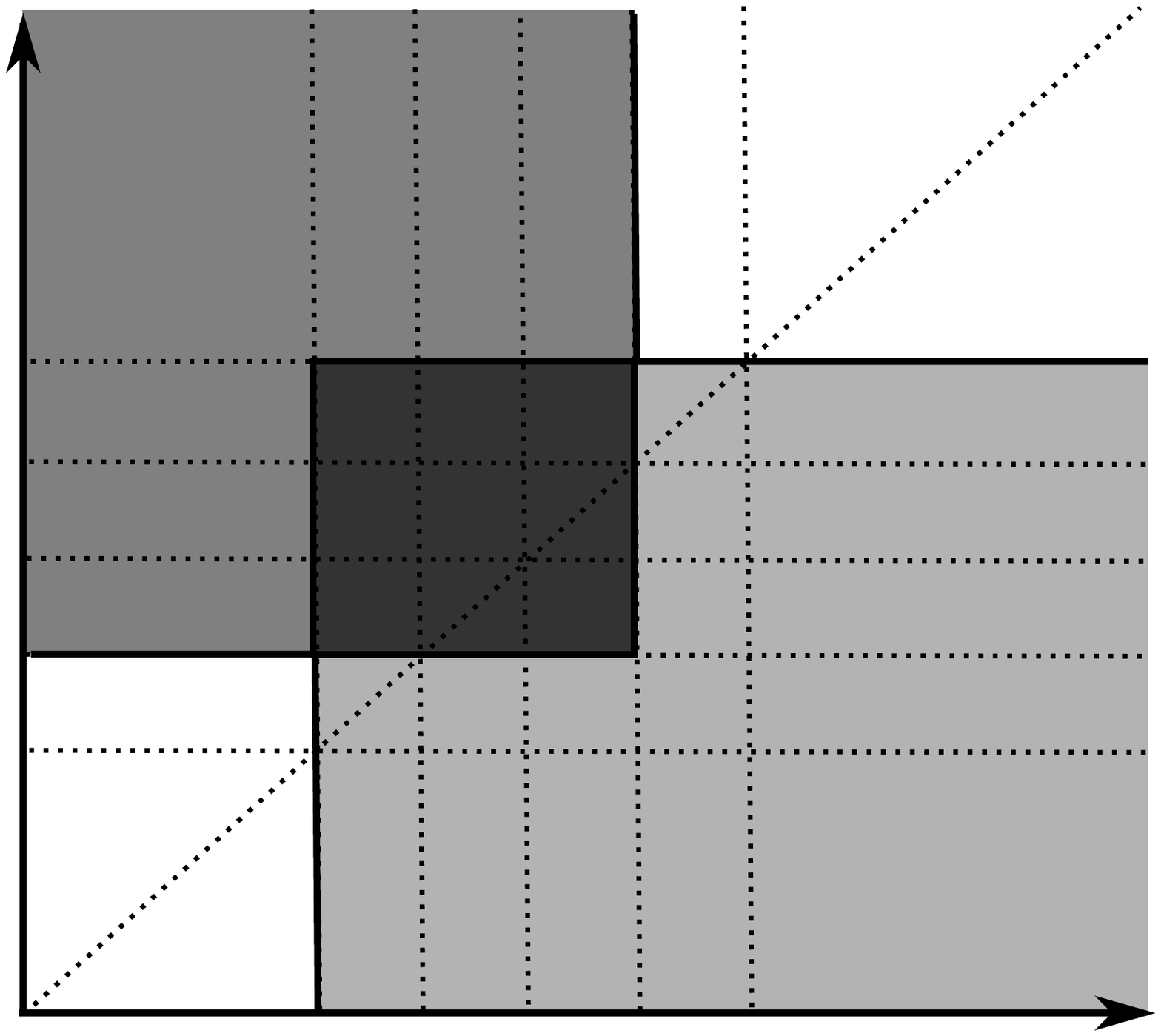
\caption{Nash equilibrium action profiles as a function of the channel ratios $\frac{g_{11}}{g_{12}}$ and $\frac{g_{21}}{g_{22}}$ for the two-player-two-channel game $\gameone$ (left) and $\gametwo$ (right), respectively. The function $\psi: \mathds{R}_+ \rightarrow \mathds{R}_+$ is defined as follows: $\psi(x) = 1 + \SNR  \, x$. The best response function $\BR_k(\bs{p}_{-k})$, for all $k \in \mathcal{K}$, is defined by (\ref{EqNEGameA}). Here, it has been arbitrarly assumed that $\psi({g_{11}}) < \psi({g_{21}})$.}
\label{FigNashRegionsG2}
\end{figure}

\begin{figure}
\centering
\def\svgwidth{1.2\linewidth}
\graphicspath{{Figures/}}
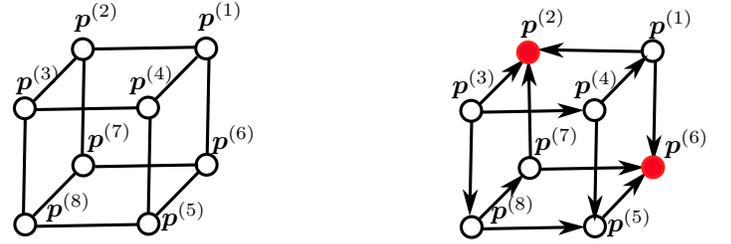
\caption{(a) Non-oriented graph and (b) oriented graph representing the game $\gametwo$ with $K = 3$, $S = 2$, under the condition $\phi(\bs{p}^{(2)}) > \phi(\bs{p}^{(6)}) > \phi(\bs{p}^{(1)}) > \phi(\bs{p}^{(5)}) > \phi(\bs{p}^{(4)}) > \phi(\bs{p}^{(7)}) > \phi(\bs{p}^{(8)}) > \phi(\bs{p}^{(3)})$. Total number of vertices: $S^{K} = 8$, number of neighbors per vertex: $K(S-1) = 3$. Maximum Number of NE: $4$. Number of NE: 2 (red vertices in (b)).}
\label{FigGraph}
\end{figure}

\begin{figure}
\centering
\includegraphics[width=\linewidth]{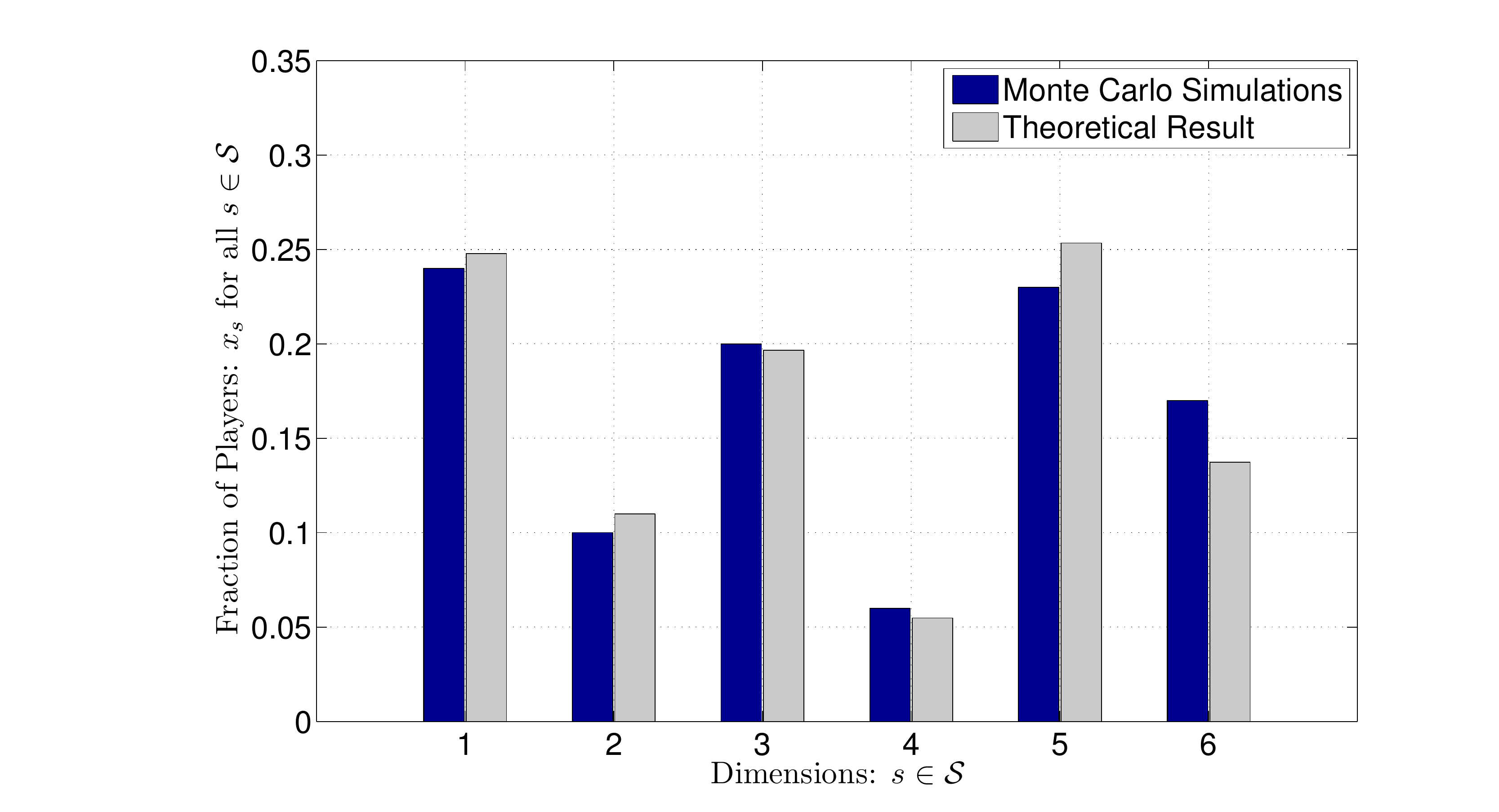}
\caption{Fraction of players transmitting over channel $s$, with $s \in \mathcal{S}$, calculated using Monte-Carlo simulations and using Eq. (\ref{EqFractionSolution}) for a network with $S = 6$ channels,  $K = 60$ players, with $\bs{b} = (b_s)_{\forall s \in \mathcal{S}} =\left( 0.25, 0.11, 0.20, 0.05, 0.25, 0.14 \right)$, and $\SNR = 10\log_{10}\left(\frac{p_{\max}}{N_0 B}\right) = 10$ dB.}
\label{FigFractions}
\end{figure}

%
\begin{figure}
\centering
\includegraphics[width=\linewidth]{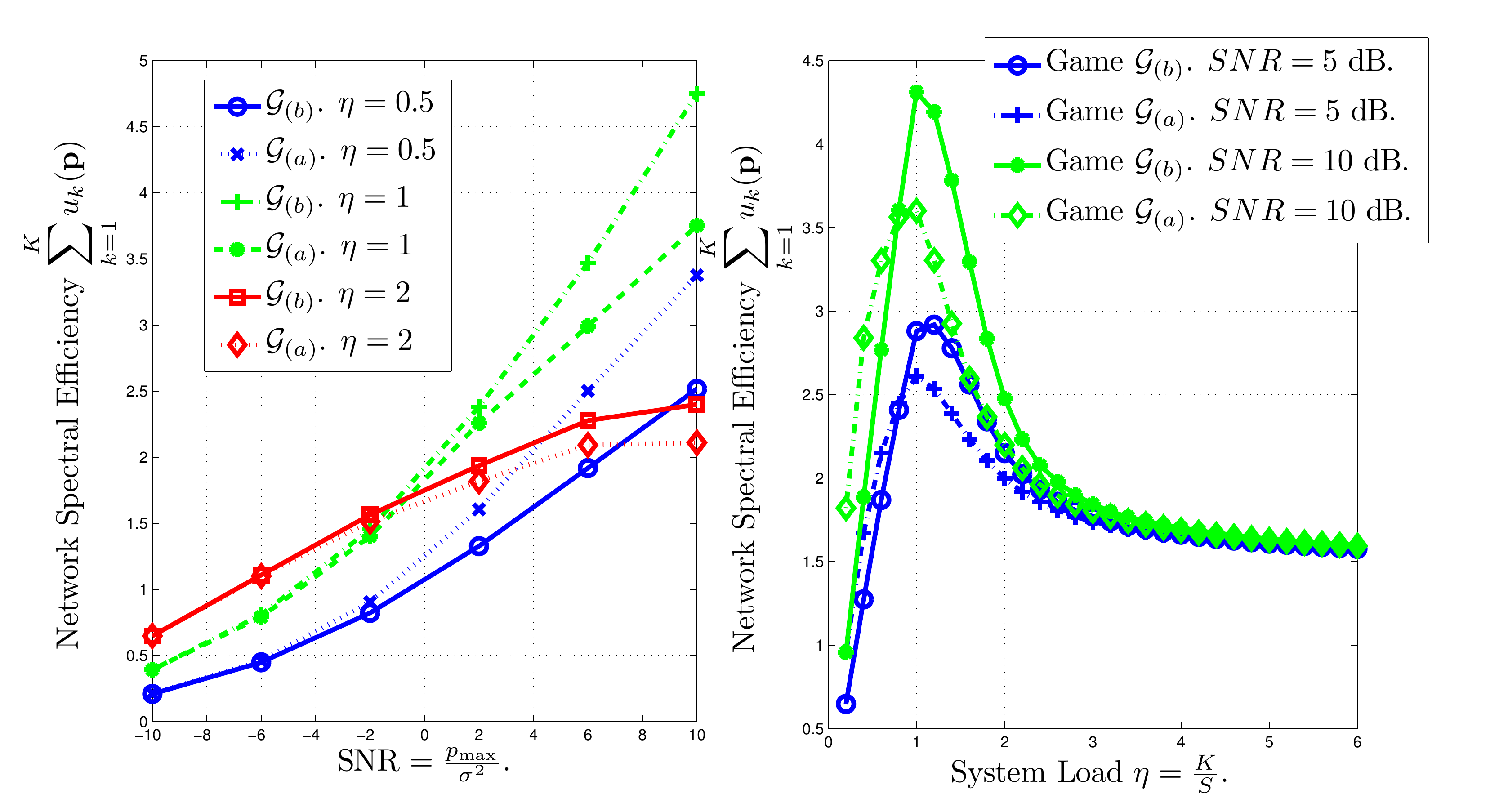}
\caption{(a) Network spectral efficiency as a function of the $\SNR = \frac{p_{\max}}{\sigma^2}$ in dBs. for the case of $\eta = \frac{K}{S} \in \lbrace\frac{1}{2}, 1 , \frac{3}{2}\rbrace$, with $K = 10$. (b) Network spectral efficiency as a function of the system load $\eta = \frac{K}{S}$ for different $\SNR = \frac{p_{\max}}{\sigma^2}$ levels in dBs. }
\label{FigLoadImpactAndSNRImpact}
\end{figure}


\begin{figure}
\centering
\includegraphics[width=\linewidth]{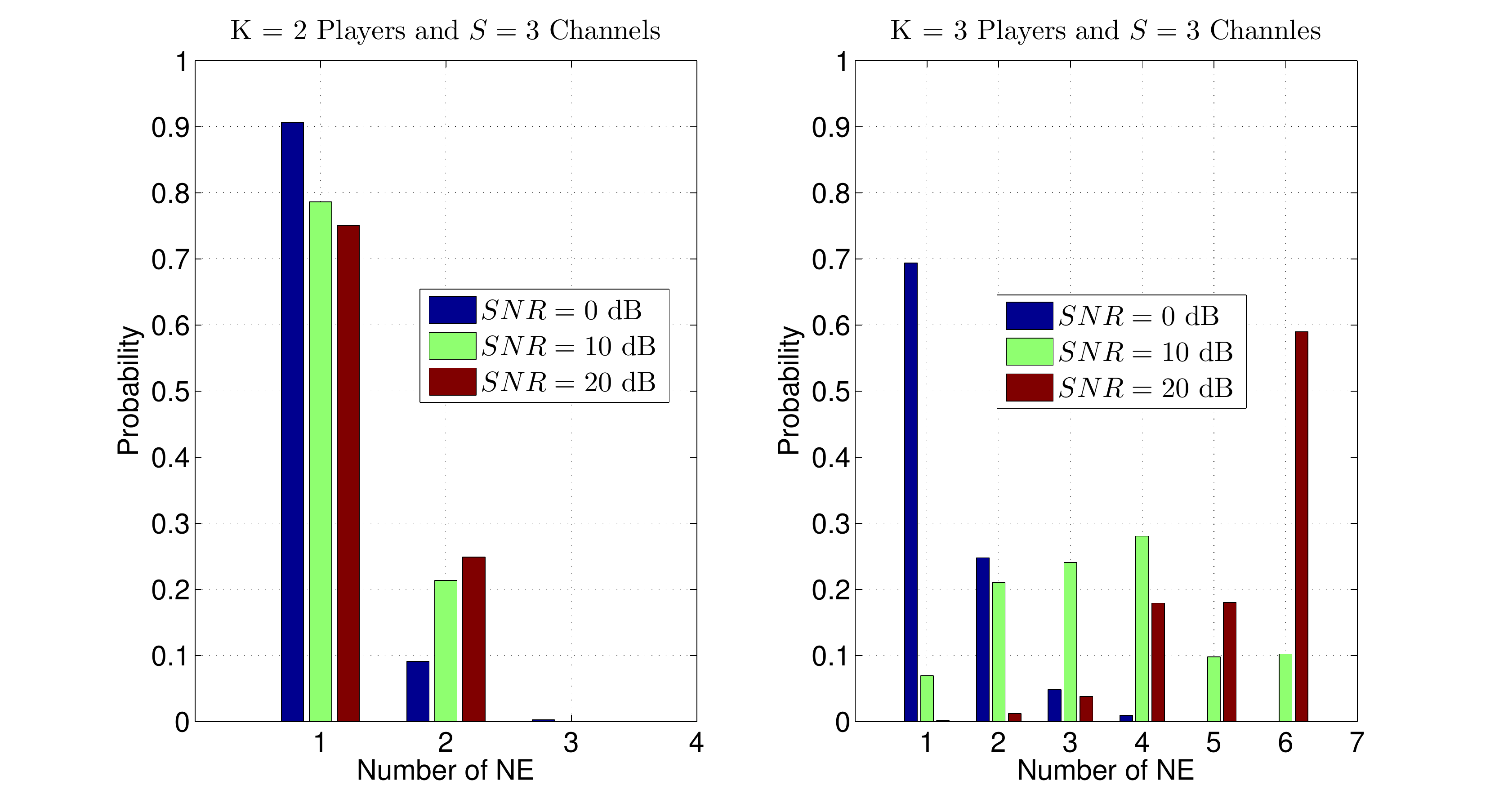}
\caption{Probability of observing a specific number of NE in the game $\gametwo$.}
\label{FigNumberOfNE}
\end{figure}

\end{document}